\documentclass[12pt,a4paper]{article}
\pdfoutput=1
\usepackage[utf8]{inputenc}
\usepackage{fullpage}
\usepackage{cite}
\usepackage{caption}
\usepackage{subcaption}
\usepackage[english]{babel}
\usepackage{amsmath}
\usepackage{physics}
\usepackage{amsthm}
\usepackage{amsfonts}
\usepackage{algorithm}
\usepackage[noend]{algpseudocode}
\usepackage{amssymb}
\usepackage{mathrsfs}
\usepackage{mathtools}
\usepackage[toc,page]{appendix}
\usepackage{graphicx}
\usepackage{float}
\usepackage{commath}
\usepackage{enumitem}
\usepackage[affil-it]{authblk}
\usepackage{siunitx}
\usepackage{braket}
\usepackage{multicol}
\usepackage{tabularx}
\usepackage{layout}
\usepackage{xcolor}

\usepackage{jheppub}
\usepackage{tikz}
\usetikzlibrary{decorations.markings, arrows.meta, calc, shapes.geometric, arrows, shapes, positioning, fit}
\usepackage{graphicx}
\usepackage{quantikz}
\pgfdeclarelayer{background}
\pgfdeclarelayer{foreground}
\pgfsetlayers{background,main,foreground}
\usepackage{environ}
\NewEnviron{myequation}{%
\begin{equation*}
\scalebox{1.7}{$\BODY$}
\end{equation*}
}

\newtheorem{proposition}{Proposition}

\def\ket#1{{|{#1}\rangle}}

\title{Magic and Non-Clifford Gates in Topological Quantum Field Theory}

\author[a]{William Munizzi,}
\author[b]{Howard J. Schnitzer}

\affiliation[a]{Division of Physical Sciences, College of Letters and Science, University of California, Los Angeles, Los Angeles, CA, 90095, USA}
\affiliation[b]{Martin Fisher School of Physics, Brandeis University, Waltham, MA 02453, USA}

\emailAdd{wmunizzi17@ucla.edu}
\emailAdd{schnitzr@brandeis.edu}

\abstract{Non-Clifford gates, used to generate quantum magic, are essential for universal quantum computation. We show that non-Clifford gates arise naturally from path integrals in topological quantum field theories, where their magic-generating properties are determined by the algebraic data of the theory. In Chern-Simons theory, we construct the Ising interaction gate, whose generator is prepared by path integration over simple three-boundary manifolds, and show that it produces non-local magic away from discrete Clifford points. We show that the Toffoli gate is obstructed in $SU(2)_1$ by the $\mathbb{Z}_2$ fusion structure, while $SU(2)_3$ is the minimal theory supporting the required conditional logic, given the density of the mapping class group in the projective unitary group on the manifold boundary. Finally, we demonstrate that the T gate arises as a path integral in Dijkgraaf-Witten theory, with gauge group $\mathbb{Z}_4$, where a single Dehn twist on the boundary torus produces the gate without approximation. These results show that topological path integrals construct gates in multiple levels of the Clifford hierarchy, and across distinct classes of field theories, with implications for topological quantum computing.}

\begin{document}
\maketitle
\flushbottom

\section{Introduction}

Quantum computation beyond the efficiently simulable Clifford group requires magic, a resource that separates universal quantum computation from classical computing~\cite{Gottesman1997,Veitch2013,Bravyi2004,Howard2014,Bravyi2016,Campbell2012,Campbell2017}. The characterization of magic, and its interplay with additional resources such as entanglement, has become a central problem in quantum resource theory with implications for fault-tolerant quantum computing, many-body physics, and holography~\cite{Leone:2021rzd,Cao:2024nrx,tirrito:2024,haug:2025,faidon:2026,Bao:2022mkc}. At the same time, the entanglement structure of stabilizer states~\cite{garcia2017geometry,Linden:2013kal,Bao2020,Keeler:2022ajf,Keeler:2023xcx,Varikuti:2025poy,Nezami:2016zni,Fuentes:2025kwg} and Clifford operators~\cite{Keeler:2023shl,Munizzi:2023ihc,Munizzi:2024huw,CuiGottesmanKrishna2017,Farinholt_2014,Selinger2013,Jagannathan:2010sb} has been studied extensively, and topological quantum field theory has proven especially well-suited to this program, with recent work characterizing the multipartite entanglement of states prepared by path integration~\cite{Salton:2016qpp,Bao2022a,Balasubramanian:2025kaf,Cummings:2025zfe,Munizzi:2025suf}. However, the topological origin of magic itself, which fundamentally lies beyond the Clifford framework, remains largely unexplored. Understanding how magic originates from the geometric and algebraic data of topological field theories would address foundational questions about the physics of quantum computational power, and provide new tools for constructing and characterizing magic states and gates.

Topological quantum field theories provide a natural framework for preparing quantum states with controlled entanglement configurations. In $U(1)$ Chern-Simons theory, at suitable levels, states prepared by path integration over three-manifolds coincide exactly with the stabilizer states~\cite{Salton:2016qpp}, which can possess maximal entanglement, but zero magic. For non-abelian theories, e.g. $SO(3)$ Chern-Simons, state preparation becomes universal and any state can be approximated to arbitrary precision. The construction of Pauli and Clifford operators as path integrals in $SU(2)_1$ Chern-Simons theory~\cite{Schnitzer:2020tiv,Munizzi:2025suf} demonstrates that fusion tensors and modular transformations directly realize Clifford group elements, providing a topological interpretation of Clifford dynamics. The magic content of certain topologically prepared states has also been quantified~\cite{Fliss:2021}, demonstrating that knot and link states in $SU(2)_k$ are generically magical and the magic of link states is predominantly long-range in nature.

Our work extends this program of realizing topological quantum operations from Clifford gates to magic gates. We construct non-Clifford unitaries as path integrals in topological quantum field theories, characterize their resource-theoretic properties, and identify the topological obstructions that constrain which gates are accessible at each level of the theory.

Our first result focuses on the Ising interaction gate $U(\theta) = \exp(-i\theta/2 (X_1 \otimes X_2))$ in $SU(2)_1$ Chern-Simons theory, a gate capable of producing non-local magic. The generator $X_1 \otimes X_2$ is prepared by path integration over a disjoint union of handlebodies, relying on the factorization property of path integrals over disconnected manifolds. We establish that $U(\theta)$ generates non-local magic for all $\theta$ away from the select Clifford points $\theta \in \{0, \pi/2, \pi\}$, and compute its non-stabilizing power $m_p(\theta) = 1/2(\sin^2(2\theta))$. The non-local character of magic generated is proven using the operator entanglement criterion~\cite{faidon:2026}, i.e. showing that conjugation by $U(\theta)$ mixes the local Pauli $Z_1 \otimes \mathbb{I}_2$ into a genuinely bipartite operator with operator linear entropy $E_{\mathrm{lin}}(U^\dagger P U) = 1/2(\sin^2(2\theta)) > 0$.

Our second result identifies a fundamental obstruction to realizing the three-qubit Toffoli gate in $SU(2)_1$. The $\mathbb{Z}_2$ fusion structure can only distinguish configurations by their parity, and therefore cannot implement the AND condition required by the Toffoli gate. We show that $SU(2)_3$ is the minimal level at which this obstruction is resolved, due to the  branching fusion rule $1/2 \otimes 1/2 = 0 \oplus 1$, which produces a spin-$1$ intermediate channel accessible only when both control qubits carry spin-$1/2$. Existence of the Toffoli gate at this level is established using the density of the mapping class group in the projective unitary group on the boundary~\cite{Freedman2003}. We describe the required properties of a manifold that can realize the Toffoli gate construction by path integration. The requisite surgery presentation and verification of leakage cancellation are described, and are left as open problems.

Our third result moves beyond Chern-Simons theory to Dijkgraaf-Witten theory~\cite{Dijkgraaf:1990} with finite gauge group $\mathbb{Z}_4$ and generating 3-cocycle $\omega_1 \in H^3(\mathbb{Z}_4, U(1))$. The modular $T$-matrix of this theory, restricted to a two-dimensional logical subspace, produces the T gate exactly without approximation. This construction reveals a contrast with the Chern-Simons case, that while both theories implement the modular $T$ operation through a Dehn twist on the boundary torus, the Chern-Simons modular $T$ produces a Clifford gate at the second level of the Clifford hierarchy, whereas the Dijkgraaf-Witten modular $T$ produces a non-Clifford gate at the third level, with the difference controlled by the cocycle data.

The results of this work establish topological quantum field theory as a setting for magic resource theory, complementing existing work on magic in Chern-Simons link states~\cite{Fliss:2021} and extending the stabilizer program~\cite{Salton:2016qpp,Cummings:2025zfe,Balasubramanian:2025kaf,Munizzi:2025suf} to the non-stabilizer regime. Our constructions provide explicit path-integral realizations of magic gates, connect their resource-theoretic properties to the algebraic data of the underlying topological quantum field theory, and identify level-dependent obstructions that constrain the magic content accessible at each level.

The paper is organized as follows. Section~\ref{Background} reviews the relevant background on topological state preparation, stabilizer states, quantum magic, and group cohomology. Section~\ref{sec:magic} constructs the Ising interaction gate in $SU(2)_1$, establishes its non-local magic content, and computes its non-stabilizing power. Section~\ref{sec:su2k} identifies the obstruction to the Toffoli gate in $SU(2)_1$, describes its realization in $SU(2)_3$, and states the associated open problems. Section~\ref{sec:DW} constructs the T gate in Dijkgraaf-Witten theory and compares the construction with the Chern-Simons results. We conclude with a discussion of implications and future directions in Section~\ref{sec:discussion}.

\section{Review}\label{Background}

In this section we provide a short review of relevant material regarding topological state preparation, stabilizer states, quantum magic, and group cohomology. Readers familiar with these topics may skip to Section~\ref{sec:magic}.

\subsection{Topological State Preparation via Path Integrals}\label{TopologicalStatePrep}

A powerful result in Chern-Simons theory surrounds the expression of the Euclidean path integral as a map from manifolds to quantum states~\cite{Witten1989,Salton:2016qpp,Dong2008,Wang2010}. For a three-manifold $\mathcal{M}$ whose boundary $\partial \mathcal{M}$ consists of $n$ torus components $\Sigma_i = \mathbb{T}^2$, path integration on $\mathcal{M}$ prepares a multipartite quantum state $\ket{\mathcal{M}} \in \mathcal{H}_{\mathbb{T}^2}^{\otimes n}$. The Hilbert space $\mathcal{H}_{\mathbb{T}^2}$ is associated to a single torus boundary, and is spanned by a finite set of basis states $\{\ket{j}\}_{j=0}^{d-1}$, one for each irreducible representation $R_j$ allowed at level $k$, where $d = \dim\mathcal{H}_{\mathbb{T}^2}$ denotes the number of integrable representations (for $U(1)$ theory, $d = k$). The basis states themselves are prepared by path integration over a solid torus $\mathcal{T}$, with a Wilson loop operator
\begin{equation}\label{WilsonLoop}
W(C, R_j) = \mathrm{Tr}_{R_j}\left[\mathcal{P} \exp\left(i \oint_C A \right)\right],
\end{equation}
inserted along a curve $C$ passing through its non-contractible core, as shown in Figure~\ref{SingleTorus}. In Eq.\ \eqref{WilsonLoop}, $A$ is a connection $1$-form serving the role of the gauge field, $\mathrm{Tr}_{R_j}$ denotes the trace taken in the representation $R_j$, and $\mathcal{P}$ provides a path ordering along $C$.
\begin{figure}[h]
\centering
\includegraphics[width=6cm]{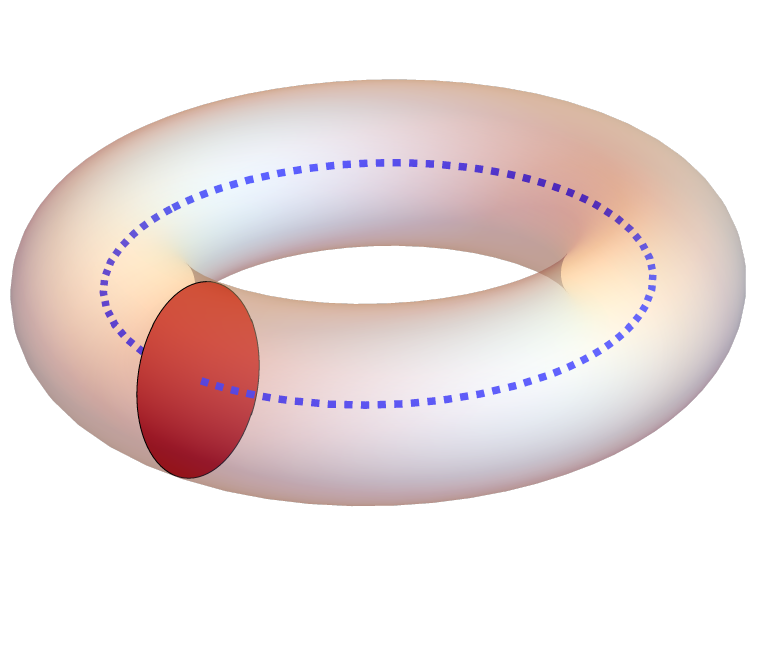}
\caption{Solid torus with a Wilson loop along a curve $C$ through its non-contractible core. Path integration along $C$ prepares the state $\ket{j}$, corresponding to the irreducible representation $R_j$ of the gauge group. The states $\{\ket{j}\}$ form an orthonormal basis for $\mathcal{H}_{\mathbb{T}^2}$.}
\label{SingleTorus}
\end{figure}

For Wilson loop operators with support on a link $L \equiv \bigcup_{i=1}^n C_i$, the corresponding observable is
\begin{equation}\label{WilsonLoopUnion}
W(L,R) = W(C_1,R_1)\cdots W(C_n,R_n).
\end{equation}
When the link $L$ may be thickened to solid torus, the expectation value of $W(L,R)$ becomes
\begin{equation}\label{WilsonLoopFraction}
\langle W(L,R) \rangle_{\mathcal{M}} = \frac{Z(\mathcal{M};L,R)}{Z(\mathcal{M})},
\end{equation}
where $Z(\mathcal{M})$ denotes the partition function on $\mathcal{M}$
\begin{equation}\label{PartitionFunction}
Z(\mathcal{M}) \equiv \int \mathcal{D}A\, e^{iS_{CS}},
\end{equation}
with $S_{CS}$ the Chern-Simons action. The function $Z(\mathcal{M};L,R)$, given by
\begin{equation}\label{PartitionFunctionWilson}
Z(\mathcal{M};L,R) = \int \mathcal{D}A\, e^{iS_{CS}} W(L,R),
\end{equation}
is the partition function after Wilson loop insertions.

Let $\mathcal{M}$ be a manifold with $n$ torus boundaries. The closed manifold $\overline{\mathcal{M}}$ is obtained by gluing $n$ solid tori, each with core $C_i$ carrying a Wilson loop in the dual representation $R_{j_i}^*$. Path integration over $\overline{\mathcal{M}}$ gives the amplitudes~\cite{Salton:2016qpp} of the prepared state $\ket{\mathcal{M}}$ as
\begin{equation}\label{StateAmplitudes}
\bra{j_1, \ldots, j_n}\ket{\mathcal{M}} = \langle W(C_1, R_{j_1}^*) \cdots W(C_n, R_{j_n}^*) \rangle_{\overline{\mathcal{M}}}.
\end{equation}
Following Eq.\ \eqref{StateAmplitudes}, every choice of three-manifold $\mathcal{M}$, with $n$ torus boundaries, determines a quantum state whose entanglement structure is encoded in the topology of $\mathcal{M}$. For a many-torus subsystem $A \subseteq \partial\mathcal{M}$, which corresponds to a stabilizer state possessing flat entanglement spectrum, the entanglement entropy $S(A)$ can be computed~\cite{Salton:2016qpp} using a single replica
\begin{equation}\label{TopologicalEE}
S(A) = -\log\left[\frac{Z(-2\mathcal{M} \cup_{f_A} 2\mathcal{M})}{Z(-\mathcal{M} \cup_{\partial \mathcal{M}} \mathcal{M})^2}\right],
\end{equation}
where $2\mathcal{M} = \mathcal{M} \cup \mathcal{M}$ indicates two copies of $\mathcal{M}$ that are used to prepare the state, and $f_A$ is an exchange diffeomorphism on the tori in $A$ that leaves $\partial\mathcal{M}\setminus A$ invariant.

\subsection{The Modular $S$ and $T$ Transformations}\label{ModularSection}

In a topological quantum field theory, e.g. Chern-Simons theory, any orientation-preserving diffeomorphism of a boundary surface induces a unitary operator on the associated Hilbert space. Since the theory is topological, isotopic diffeomorphisms yield the same operator up to phase, so the resulting projective representation factors through the mapping class group of the surface. For the torus $\mathbb{T}^2$, the mapping class group is isomorphic to $\mathrm{SL}(2,\mathbb{Z})$, generated by the modular $S$ and $T$ transformations, which play a distinguished role in topological state preparation.

The topological action of the modular $T$ transformation corresponds to a Dehn twist, where the torus is cut along a non-contractible cycle, one side is rotated by $2\pi$, and the two open ends are then glued back together. On the Hilbert space $\mathcal{H}_{\mathbb{T}^2}$, this operation is represented by the diagonal unitary
\begin{equation}\label{ModularT}
T_{jj'} = \delta_{j,j'}\, e^{2\pi i(h_j - c/24)},
\end{equation}
where $h_j$ indicates the conformal dimension of representation $j$, and $c$ is the central charge of the theory. For $SU(2)$ at level $\ell$, the torus Hilbert space has dimension $d = \ell + 1$, the central charge is $c = 3\ell/(\ell+2)$, and the conformal dimensions are $h_j = j(j+1)/(\ell+2)$, with $j \in \{0, \frac{1}{2}, \ldots, \frac{\ell}{2}\}$. At level $\ell = 1$ we have $d = 2$, $c = 1$, and $h_j = j(j+1)/3$, so the $SU(2)_1$ modular $T$ matrix becomes
\begin{equation}\label{ModularT_SU21}
T_{ab} = \exp\left[2\pi i\left(h_a - \frac{1}{24}\right)\right]\delta_{ab},
\end{equation}
and acts as a diagonal phase gate on $\mathcal{H}_{\mathbb{T}^2}$. The standard Clifford phase gate $P$ is then constructed from $T_{ab}$ as
\begin{equation}\label{PhaseGate}
P = Z^{(d-1)/2}T,
\end{equation}
with $Z$ the Pauli $Z$ operator and $(d-1)/2$ an integer for odd $d$.

The modular $S$ transformation exchanges the two fundamental cycles of the torus. In the Hilbert space $\mathcal{H}_{\mathbb{T}^2}$, $S$ is represented by the unitary
\begin{equation}\label{ModularS}
S_{jj'} = \frac{1}{\sqrt{d}}\,\omega^{jj'}, \qquad \omega = e^{2\pi i/d},
\end{equation}
where $j, j' \in \{0, \ldots, d-1\}$ label the basis states of $\mathcal{H}_{\mathbb{T}^2}$ and $d = \dim\mathcal{H}_{\mathbb{T}^2}$ is the number of integrable representations. The $S$ matrix in Eq.\ \eqref{ModularS} implements a discrete Fourier transform on $\mathcal{H}_{\mathbb{T}^2}$, transforming from the computational basis $\{\ket{j}\}$. For $SU(2)_1$, with $d = 2$, Eq.\ \eqref{ModularS} reduces to the Hadamard matrix
\begin{equation}\label{HadamardMatrix}
S_{ab} = \frac{1}{\sqrt{2}}\begin{bmatrix} 1 & 1\\ 1 & -1 \end{bmatrix}.
\end{equation}
The operators $S$ and $T$ generate the full action of $\mathrm{SL}(2,\mathbb{Z})$ on $\mathcal{H}_{\mathbb{T}^2}$, and serve as the primary building blocks for the topological construction of Clifford operators in Section~\ref{StabilizerTopology}.

\subsection{The Fusion Tensor and Pauli Operators}\label{FusionTensorSection}

In Chern-Simons theory with gauge group $G$ and level $k$, the fusion tensor $N^{j_3}_{j_1 j_2}$ counts the independent ways representations $j_1$ and $j_2$ can fuse to give representation $j_3$. The identity
\begin{equation}\label{FusionIdentity}
    N^0_{j_1 j_2} = \delta_{j_1, j_2^*}
\end{equation}
ensures that only conjugate representations fuse to the trivial channel. An important manifold~\cite{Salton:2016qpp} for entangled state preparation is the three-boundary manifold $\eta$, defined as a solid torus with two solid tori removed from its interior (Figure~\ref{Eta}).
\begin{figure}[h]
\centering
\includegraphics[width=6cm]{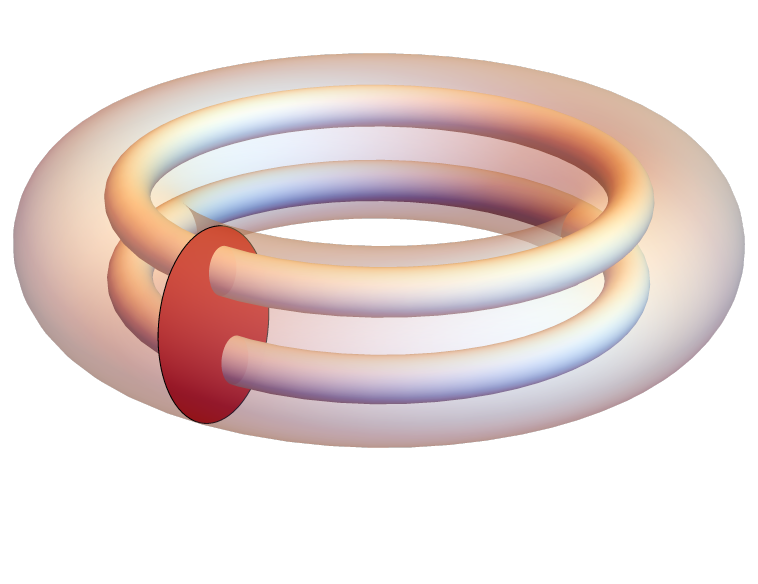}
\caption{The manifold $\eta$, consisting of a solid torus with two tori removed from its interior. Path integration over $\eta$, with Wilson loops through the interior tori, prepares a state whose amplitudes are proportional to $N^{j_3}_{j_1 j_2}$.}
\label{Eta}
\end{figure}
Path integration over $\eta$, with Wilson loops inserted through the interior tori, prepares a tripartite state with amplitudes proportional to $N^{j_3}_{j_1 j_2}$:
\begin{equation}\label{EtaState}
\ket{\eta} \propto \sum_{j_1, j_2, j_3} N^{j_3}_{j_1 j_2} \ket{j_1, j_2, j_3}.
\end{equation}

For $U(1)$ Chern-Simons theory at odd level $k$, the fusion tensor takes the form
\begin{equation}\label{U1FusionTensor}
N^{j_3}_{j_1 j_2} = \delta_{j_3,\, j_1 + j_2 \pmod{k}},
\end{equation}
with $j_1, j_2, j_3 \in \{0, \ldots, k-1\}$, reflecting $U(1)$ charge conservation modulo $k$. Viewed as a map between Hilbert spaces, $N^{j_3}_{j_1 j_2}$ acts as a copy operation in the Fourier-transformed basis after conjugation by the modular $S$ matrix. Viewed as a tripartite state, it describes a GHZ state, providing a path integral method for generating multipartite entanglement.

The fusion tensor further enables an algebraic realization of the Pauli $X$ and $Z$ operators. The Kac-Moody algebra $SU(2)_1$ admits two integrable representations, spin-$0$ and spin-$\frac{1}{2}$, which we label $a=0$ and $a=1$ respectively. Fixing one index of $N^{a_3}_{a_1 a_2}$ to $a_2 = 1$ and defining the operator $X$ by its matrix elements gives
\begin{equation}\label{XMatrixElements}
\bra{a_2} X \ket{a_1} \equiv N^{a_2}_{a_1,\, 1} = \delta_{a_2,\, a_1 + 1 \pmod{2}}.
\end{equation}
The operator $X$, on a state $\ket{a}$, acts as the qubit shift operator
\begin{equation}\label{ShiftFromFusion}
X\ket{a} = \sum_{a_2} \delta_{a_2,\, a + 1 \pmod{2}} \ket{a_2} = \ket{a + 1 \pmod{2}}.
\end{equation}
The Pauli $Z$ operator $Z\ket{a} = \omega^a \ket{a}$, with $\omega = e^{2\pi i/d}$, is obtained by conjugation $Z = SXS^\dagger$. The operators $X$ and $Z$ satisfy the standard commutation relation $XZ = \omega^{-1}ZX$, and are defined independently of the level. For $SU(2)_1$, where $d = 2$, we have $\omega = -1$. A graphical representation of the tensor $N^{a_2}_{a_1,\, 1}$ is given in Figure~\ref{TensorGraph}.
\begin{figure}[h]
\begin{center}
\begin{tikzpicture}
\coordinate (O) at (0,0);
\coordinate (A) at (2,-1);
\coordinate (B) at (-2,-1);
\coordinate (C) at (0,2);
\draw[thick] (O) -- (A);
\draw[thick] (O) -- (B);
\draw[thick] (O) -- (C);
\node at (-4.2,.5) {$N_{a,\,1}^{a_2} = \delta_{a_2,\, a+1 \pmod{2}} = $};
\node at (2.2,-1.2) {$1$};
\node at (-2.2,-1.2) {$a$};
\node at (0.5,2.2) {$a+1\pmod{2}$};
\draw[thick,->] (A) -- ($(O)!0.5!(A)$);
\draw[thick,->] (B) -- ($(O)!0.5!(B)$);
\draw[thick,->] (O) -- ($(O)!0.5!(C)$);
\end{tikzpicture}
\caption{Graphical representation of $N_{a,\,1}^{a_2} = \delta_{a_2,\, a+1\pmod{2}}$ for $SU(2)_1$. Two incoming legs are labeled $a$ and $1$, and the single outgoing leg gives the fusion product $a+1 \pmod{2}$. The tensor is prepared by path integration over $\eta$, Figure~\ref{Eta}, with Wilson loops.}
\label{TensorGraph}
\end{center}
\end{figure}
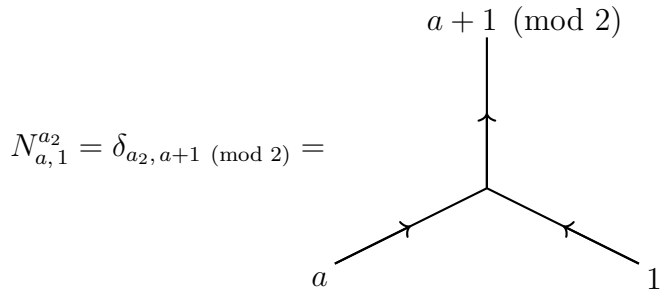

\subsection{The Clifford Group and Stabilizer States}\label{StabilizerTopology}

The Clifford group $\mathcal{C}_n$ is defined as the normalizer~\cite{NielsenChuang2010} of the $n$-qubit Pauli group $\Pi_n$, i.e. the set of unitaries that map $\Pi_n$ to itself under conjugation. Clifford operations are efficiently simulable on a classical computer~\cite{aaronson2004improved,Gottesman1997,Gottesman1998}, and the set of $n$-qubit stabilizer states are exactly those states obtainable by Clifford group action on the computational basis. For qubits, the group $\mathcal{C}_n$ is generated by the Hadamard $H$, phase gate $P$, and controlled-NOT $C_{i,j}$ gates, with respective matrix representations
\begin{equation}\label{CliffordGates}
    H\equiv \frac{1}{\sqrt{2}}\begin{bmatrix}1&1\\1&-1\end{bmatrix}, \quad P\equiv \begin{bmatrix}1&0\\0&i\end{bmatrix}, \quad C_{i,j} \equiv \begin{bmatrix}1 & 0 & 0 & 0\\0 & 1 & 0 & 0\\0 & 0 & 0 & 1\\0 & 0 & 1 & 0\end{bmatrix}.
\end{equation}

Each of the qubit Clifford generators admits a $d$-dimensional generalization that can be realized topologically~\cite{Gross2006,Hostens2005}, where $d = \dim\mathcal{H}_{\mathbb{T}^2}$. The Hadamard gate generalizes to the modular $S$ transformation, given in Eq.\ \eqref{ModularS}. The action of the Hadamard operator transforms Paulis as $X \mapsto Z$, $Z \mapsto X^{-1}$, and is realized topologically by a Dehn twist on the torus boundary, as described in Section~\ref{ModularSection}. Analogously the phase gate is constructed through the action of modular $T$ (Eq.\ \eqref{ModularT}) with Pauli $Z$, as in Eq.\ \eqref{PhaseGate}, and implements the map $X \mapsto XZ$, $Z \mapsto Z$.

The controlled-sum gate $C_{\mathrm{sum}}$ serves as the $d$-dimensional extension of the CNOT gate, and acts on a pair of qudits $i$ and $j$ as
\begin{equation}\label{CSum}
C_{\mathrm{sum}}\ket{i}\ket{j} = \ket{i}\ket{i+j \pmod{d}},
\end{equation}
for odd dimension $d$. For $d$ even, an additional phase is required, such that
\begin{equation}\label{CSumEven}
C_{\mathrm{sum}}\ket{i}\ket{j} = \omega^{\frac{1}{2}(i+j)}\ket{i}\ket{i+j \pmod{d}},
\end{equation}
where $\omega = e^{2\pi i/d}$ as before. The action of $C_{\mathrm{sum}}$ on the Pauli operators maps
\begin{equation}\label{CSumMapping}
\begin{split}
X \otimes I &\mapsto X \otimes X, \qquad I \otimes X \mapsto I \otimes X,\\
Z \otimes I &\mapsto Z \otimes I, \qquad I \otimes Z \mapsto Z^{-1} \otimes Z.
\end{split}
\end{equation}

Topologically, $C_{\mathrm{sum}}$ is constructed by first forming a copy tensor from the adjoint of the fusion tensor $N^{j_3}_{j_1 j_2}$, conjugated by $S$. This copy tensor is then contracted with the fusion tensor. For odd $d$, Figure~\ref{CSumGraph} illustrates the graphical representation of $C_{\mathrm{sum}}$.
\begin{figure}[h]
\begin{center}
\begin{tikzpicture}[thick, scale=1]
\draw (0,0) -- (0,0.5);
\draw (0,1.1) -- (0,2);
\draw (0,2) -- (0,2.9);
\draw (0,3.5) -- (0,4);
\draw (3,0) -- (3,4);
\draw[->] (0,2) -- (1.3,2);
\draw[->] (1.9,2) -- (3,2);
\draw (-0.4,0.5) rectangle (0.4,1.1);
\node at (0,0.8) {\( S^\dagger \)};
\draw (-0.4,2.9) rectangle (0.4,3.5);
\node at (0,3.2) {\( S \)};
\draw (1.3,1.6) rectangle (1.9,2.4);
\node at (1.6,2) {\( S \)};
\filldraw (0,2) circle (2pt);
\filldraw (3,2) circle (2pt);
\draw[->] (0,-0.5) -- (0,0);  \node at (0,-0.7) {\( j \)};
\draw[->] (3,-0.5) -- (3,0);  \node at (3,-0.7) {\( i \)};
\draw[->] (0,4) -- (0,4.5);  \node at (0,4.7) {\( j \)};
\draw[->] (3,4) -- (3,4.5);  \node at (3,4.7) {\( i + j \)};
\node at (.75,1.7) {\( j \)};
\node at (-0.3,2.1) {\( N \)};
\node at (3.3,2.1) {\( N \)};
\end{tikzpicture}
\caption{Graphical representation of $C_{\mathrm{sum}}$ for odd $d$, built by contracting $N^{j_3}_{j_1 j_2}$ with its copy tensor. Modular $S$ is applied on each outgoing leg, with $S^\dagger$ on the incoming leg. Incoming legs have control label and target $j$, while outgoing legs have labels $j$ and $i+j$.}
\label{CSumGraph}
\end{center}
\end{figure}
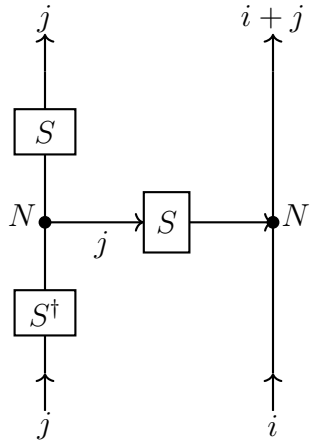
For even dimension $d$, Figure~\ref{CSumPhase} gives the graphical representation of the $C_{\mathrm{sum}}$ operator, with associated phase gate $P$ applied to the outgoing leg.
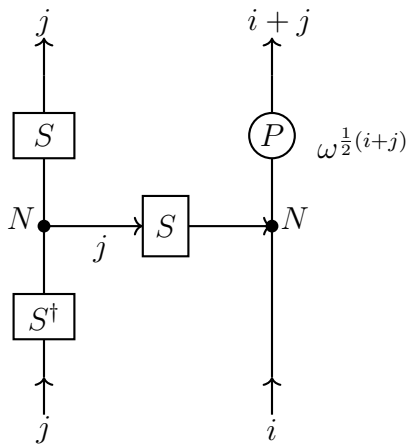
\begin{figure}[h]
\begin{center}
\begin{tikzpicture}[thick, scale=1]
\draw (0,0) -- (0,0.5);
\draw (0,1.1) -- (0,2);
\draw (0,2) -- (0,2.9);
\draw (0,3.5) -- (0,4);
\draw[->] (0,2) -- (1.3,2);
\draw[->] (1.9,2) -- (3,2);
\draw (-0.4,0.5) rectangle (0.4,1.1);
\node at (0,0.8) {\( S^\dagger \)};
\draw (-0.4,2.9) rectangle (0.4,3.5);
\node at (0,3.2) {\( S \)};
\draw (1.3,1.6) rectangle (1.9,2.4);
\node at (1.6,2) {\( S \)};
\draw (3,3.2) circle (0.3);
\node at (3,3.2) {\( P \)};
\draw (3,0) -- (3,2.9);
\draw (3,3.5) -- (3,4);
\filldraw (0,2) circle (2pt);
\filldraw (3,2) circle (2pt);
\draw[->] (0,-0.5) -- (0,0); \node at (0,-0.7) {\( j \)};
\draw[->] (3,-0.5) -- (3,0); \node at (3,-0.7) {\( i \)};
\draw[->] (0,4) -- (0,4.5); \node at (0,4.7) {\( j \)};
\draw[->] (3,4) -- (3,4.5); \node at (3.1,4.7) {\( i + j \)};
\node at (-0.3,2.1) {\( N \)};
\node at (3.3,2.1) {\( N \)};
\node at (.75,1.7) {\( j \)};
\node at (4.2,3.1) {\( \omega^{\frac{1}{2}(i + j)} \)};
\end{tikzpicture}
\caption{Graphical representation of $C_{\mathrm{sum}}$ for even $d$. An additional phase gate $P$, contributing $\omega^{\frac{1}{2}(i+j)}$, is applied to the outgoing target leg, following Eq.\ \eqref{CSumEven}.}
\label{CSumPhase}
\end{center}
\end{figure}

Since Pauli $X$ and $Z$, as well as the $S$, $P$, and $C_{\mathrm{sum}}$ operators are topologically constructible, the full $n$-qudit Clifford group can be described by path integration over manifolds. Moreover, any Clifford unitary $\mathcal{C}$ on $n$ qudits corresponds to a path integral over some three-manifold $\mathcal{M}$ with $2n$ torus boundaries~\cite{Salton:2016qpp}. Accordingly, gluing solid tori to the input boundaries then produces an arbitrary stabilizer state $\mathcal{C}\ket{0}^{\otimes n}$, as described in Theorem $1$ of ~\cite{Salton:2016qpp}.

All stabilizer states have a flat entanglement spectrum, enabling entanglement entropies $S(A)$ to be computed using a single replica, by Eq.\ \eqref{TopologicalEE}, without analytic continuation. Furthermore, the tripartite entanglement that can be distilled by local unitaries between subsystems $A$, $B$, and $C$ (with boundary tripartition $\partial\mathcal{M} = A \cup B \cup C$) can likewise be described~\cite{Salton:2016qpp,Nezami:2016zni} topologically as
\begin{equation}\label{GHZFormula}
g = \frac{S(A) + S(B) + S(C)}{\log d} + \log_d \frac{Z(-3\mathcal{M} \cup_f 3\mathcal{M})}{Z(-\mathcal{M} \cup_{\partial \mathcal{M}} \mathcal{M})^3},
\end{equation}
where $f$ cyclically permutes the three torus copies in $A$, reverses the cyclic order in $B$, and fixes tori in $C$. Eqs.\ \eqref{TopologicalEE} and\ \eqref{GHZFormula} give a complete topological characterization of the bipartite and tripartite entanglement for states prepared by path integration in $U(1)$ Chern-Simons theory.

\subsection{Quantum Magic and Non-Local Magic}\label{MagicSection}

Magic quantifies non-stabilizerness in a quantum system, i.e. the degree to which a state or process lies beyond the capabilities of the Clifford group~\cite{Veitch2013}. It is well known that magic provides a necessary resource for quantum algorithms to achieve advantage over classical computation, and the precise characterization of magic, and its associated generation by unitary operators, remains an active area of research~\cite{Leone:2021rzd,Cao:2024nrx,tirrito:2024, haug:2025, faidon:2026, Munizzi:2025suf, Khumalo:2025xfv}. One canonical measure of magic is the stabilizer R\'{e}nyi entropy (SRE)~\cite{Leone:2021rzd}. For an $n$-qubit pure state $\ket{\psi}$, the SRE of order $\alpha$ is defined
\begin{equation}\label{eq:SRE1}
M_{\alpha}\left(\psi\right) \equiv \frac{1}{1-\alpha} \ln\left(\sum_{P \in \mathcal{P}_n} \frac{\langle\psi|P|\psi\rangle^{2\alpha}}{2^n}\right) - \ln\left(2^n\right),
\end{equation}
where $\mathcal{P}_n$ is the set of $n$-qubit Pauli strings.

The non-stabilizing power $m_p(U(\theta))$, of a unitary $U$, provides an operator level analogue of SRE, measuring the average non-stabilizerness generated by $U$ when acting on stabilizer states~\cite{Varikuti:2025poy}. Let $\mathcal{M}(\ket{\psi})$ denote the linear stabilizer entropy of a pure state $\ket{\psi}$, defined
\begin{equation}\label{LinearStabEntropy}
\mathcal{M}(\ket{\psi}) = 1 - 2^n\,\mathrm{Tr}\left[Q\,(\ket{\psi}\bra{\psi})^{\otimes 4}\right],\quad \textnormal{with} \quad Q = \frac{1}{2^{2n}}\sum_{P \in \mathcal{P}_n} P^{\otimes 4},
\end{equation}
where $Q$ serves as a projector in $\mathcal{H}^{\otimes 4}$. The non-stabilizing power $m_p(U)$ of $U$ is then the average of $\mathcal{M}(U\ket{\psi})$ over all stabilizer states in the Hilbert space
\begin{equation}\label{NSPDef}
m_p(U) = \overline{\mathcal{M}(U\ket{\psi})} = 1 - 2^n\,\mathrm{Tr}\left[Q\,U^{\otimes 4}\,\overline{(\ket{\psi}\bra{\psi})^{\otimes 4}}\,U^{\dagger\otimes 4}\right].
\end{equation}
The value of $m_p(U)$ vanishes if and only if $U$ is a Clifford, and is invariant under composition with Clifford unitaries on either side.

Magic can also reside non-locally in a system, tied closely to the entanglement structure of the state. Given a magic measure $M$ and a bipartite Hilbert space $\mathcal{H} = \mathcal{H}_A \otimes \mathcal{H}_B$, the non-local magic of a state $\ket{\psi} \in \mathcal{H}$ is computed
\begin{equation}\label{eq:NLMagic}
M^{NL}(\psi) \equiv \min_{U_A \otimes U_B} M\left(U_A \otimes U_B \ket{\psi}\right),
\end{equation}
where $U_A$ and $U_B$ are arbitrary unitaries that act locally on $\mathcal{H}_A$ and $\mathcal{H}_B$ respectively. Consequently, non-local magic describes the magic across a Hilbert space bipartition, which cannot be created or destroyed by local operations. Conjugation by local unitary action cannot generate operator entanglement, therefore the operator entanglement of a conjugated Pauli string provides a proxy for non-local magic~\cite{faidon:2026}.

Given a unitary operator $U$ acting on $\mathcal{H}$, the operator linear entropy of $U$ is then~\cite{faidon:2026}
\begin{equation}\label{ElinDef}
E_{\mathrm{lin}}(U) = 1-\frac{1}{d^{2n}}\,\mathrm{Tr}\left(T^A_{12}\,U^{\otimes 2}\,T^A_{12}\,U^{\dagger\otimes 2}\right),
\end{equation}
where $T^A_{12}$ is a partial swap of the $A$ subsystems of the doubled space $\mathcal{H}_A^1 \otimes \mathcal{H}_B^1 \otimes \mathcal{H}_A^2 \otimes \mathcal{H}_B^2$. Equivalently, given the operator Schmidt decomposition
\begin{equation}\label{OpSchmidtDecomp}
    U/\sqrt{d^n} = \sum_i \sqrt{\lambda_i}\, V_i \otimes W_i, 
\end{equation}
where $\{V_i\}$ and $\{W_i\}$ are orthonormal under the Hilbert-Schmidt inner product and $\sum_i \lambda_i = 1$, the operator linear entropy reduces to
\begin{equation}\label{ELimReduction}
    E_{\mathrm{lin}}(U) = 1 - \sum_i \lambda_i^2.  
\end{equation}
The quantity $E_{\mathrm{lin}}(U)$ is zero when $U$ has Schmidt rank $1$, i.e. when $U = V \otimes W$ is a product operator. Moreover, $E_{\mathrm{lin}}(U) > 0$ when $U$ has bipartite operator structure.

Operator linear entropy serves as a probe for non-local magic~\cite{faidon:2026} through the following. For an $n$-qubit unitary $U$ acting on $\mathcal{H} = \mathcal{H}_A \otimes \mathcal{H}_B$,
\begin{equation}\label{NonLocalMagicCriterion}
E_{\mathrm{lin}}\left(U^\dagger P U \right) = 0 \quad \forall P \in \mathcal{P}_n \iff U = (V \otimes W)\,C, \quad V \in \mathcal{U}_A,\; W \in \mathcal{U}_B,
\end{equation}
where $C$ is a Clifford unitary and $\mathcal{P}_n$ is the set of all $n$-qubit Pauli strings. Otherwise stated, conjugation by $U$ preserves a product structure on every Pauli string if and only if $U$ is the product of local unitaries followed by a Clifford. Consequently, any $U$ satisfying Eq.\ \eqref{NonLocalMagicCriterion} cannot generate non-local magic. In what follows we employ $E_{\mathrm{lin}}$ in two distinct roles: $E_{\mathrm{lin}}(U)$ measures the operator entanglement of the gate $U$ itself, while $E_{\mathrm{lin}}(U^\dagger P U)$ detects whether conjugation by $U$ rotates a local Pauli into a genuinely bipartite operator, as in criterion\ \eqref{NonLocalMagicCriterion}. These two quantities have in general different functional forms and are maximized at different parameter values, as we will show in Section~\ref{sec:magic}.

\subsection{Group Cohomology}\label{subsec:cohomology}

Homology and cohomology provide useful algebraic tools for classifying the global structure of topological spaces and groups. Given a topological space, homology associates to that space a sequence of abelian groups that encode features such as connected components, loops, or defects. Dually, cohomology assigns algebraic data, e.g. functions or phases, to topological features of that space, describing topological action in gauge theory. In Section~\ref{sec:DW} we utilize cohomology to give a topological realization of the T gate, and therefore offer a short review of essential terminology and constructions here.

We restrict our attention to the cohomology of finite groups, which is the setting for the Dijkgraaf-Witten theory used in Section~\ref{sec:DW}. Given a group $G$, a $k$-cochain on $G$ is a function $\omega$ that assigns a phase $e^{i\alpha} \in U(1)$ to each ordered $k$-tuple of $g \in G$,
\begin{equation}\label{kCochain}
\omega\,\colon \underbrace{G \times G \times \cdots \times G}_{k} \to U(1).
\end{equation}
The set of all $k$-cochains forms an abelian group $C^k(G, U(1))$ under pointwise multiplication.

We further define an operator $\delta$, known as the coboundary operator, that  maps $k$-cochains to $(k+1)$-cochains (the associated $(k+1)$ coboundary). Formally, $\delta$ is defined
\begin{equation}\label{CoboundaryMap}
\delta\,\colon\, C^k(G, U(1)) \to C^{k+1}(G, U(1)).
\end{equation}
As an example, given a $2$-cochain $\sigma(a, b)$ the coboundary $\delta\sigma$ is the $3$-cochain
\begin{equation}\label{CoboundaryExample}
(\delta\sigma)(a, b, c) = \frac{\sigma(b, c)\;\sigma(a, b \cdot c)}{\sigma(a \cdot b, c)\;\sigma(a, b)},
\end{equation}
where $a, b, c \in G$ and $\cdot$ indicates the group operation of $G$. The operator $\delta$ evaluates the original cochain on subtuples of the input tuple, with alternating factors appearing in the numerator and denominator, and generalizes to arbitrary $k$ in a standard way~\cite{Dijkgraaf:1990}. An important property of $\delta$ is that applying it twice produces the trivial action
\begin{equation}\label{DeltaSquared}
\delta^2 = 1,
\end{equation}
for any cochain.

A $k$-cocycle is a $k$-cochain $\omega$ that satisfies the condition
\begin{equation}\label{CocycleCondition}
    \delta\omega = 1,
\end{equation}
and a $k$-coboundary is a $k$-cochain of the form $\omega = \delta\sigma$ for $(k-1)$-cochain $\sigma$. Following Eq.\ \eqref{CocycleCondition} every coboundary is a cocycle, but not every cocycle is a coboundary. The $k$-th cohomology group of $G$ with $U(1)$ coefficients, denoted $H^k(G, U(1))$, is then defined as the quotient
\begin{equation}\label{CohomologyGroup}
H^k(G, U(1)) = \frac{\ker\delta\,\colon\, C^k \to C^{k+1}}{\mathrm{im}\,\delta\,\colon\, C^{k-1} \to C^k} = \frac{\{k\textnormal{-cocycles}\}}{\{k\textnormal{-coboundaries}\}}.
\end{equation}
The elements of $H^k(G, U(1))$ are equivalence classes of cocycles, where two cocycles are identified together if they differ by a coboundary. Intuitively, $H^k(G, U(1))$ classifies the distinct ways to assign $U(1)$ phases to $k$-tuples of group elements that satisfy the cocycle condition, and are not removable by a redefinition of phases on $(k-1)$-tuples. We use the notation
\begin{equation}
    [\omega] \in H^k(G, U(1))
\end{equation}
to denote the equivalence class of a cocycle $\omega$ in $H^k(G, U(1))$. 

For the cyclic group $G = \mathbb{Z}_N$, the third ($k=3$) cohomology group is
\begin{equation}\label{H3ZN}
H^3(\mathbb{Z}_N, U(1)) \cong \mathbb{Z}_N.
\end{equation}
Consequently there are $N$ inequivalent $3$-cocycles on $\mathbb{Z}_N$, labeled by parameter $p \in \{0, 1, \ldots, N-1\}$. The case $p = 0$ corresponds to the trivial cocycle $\omega_0 = 1$, while the case $p = 1$ gives the generating cocycle $\omega_1$ from which all others can be obtained as powers
\begin{equation}\label{OmegaPowers}
    \omega_p = \omega_1^p.
\end{equation}
As we show in Section~\ref{sec:DW}, the generator $\omega_1$ for $G = \mathbb{Z}_4$ encodes the topological action of Dijkgraaf-Witten theory, where the modular $T$-matrix realizes the T gate.

\section{Topological Non-Local Magic in $SU(2)_1$}\label{sec:magic}

In this section we construct a one-parameter family of non-Clifford gates in $SU(2)_1$ Chern-Simons theory, and characterize their magic generating capabilities. We begin by demonstrating that the two-qubit Pauli string $X_1 \otimes X_2$ can be realized topologically by path integration over the disjoint union of $\eta$ manifolds, shown in Figure~\ref{Eta}. The operator $X_1 \otimes X_2$ serves as the generator of the Ising interaction gate $U(\theta)$, which we show produces non-local magic for all values of $\theta$ away from a select few  points using the operator entanglement criterion of Eq.\ \eqref{NonLocalMagicCriterion}. We then compute the non-stabilizing power of $U(\theta)$ by averaging over all two-qubit stabilizer states.

\subsection{Path Integral Construction of the Ising Interaction Gate in $SU(2)_1$}\label{subsec:XX_construction}

As established in Section~\ref{FusionTensorSection}, path integration over the manifold $\eta$ (Figure~\ref{Eta}) with a Wilson loop in the spin-$1/2$ representation gives a topological realization for the Pauli $X$ operator, via the fusion tensor $N^{a_2}_{a,\,1} = \delta_{a_2,\,a+1\pmod{2}}$. Constructing the $2$-qubit operator $X_1\otimes X_2$, we take the disjoint union $\eta_1 \sqcup \eta_2$, two copies of the manifold $\eta$, each possessing an independent Wilson loop in the spin-$1/2$ representation. The path integral over a disjoint union of manifolds factors into the tensor product of the individual path integrals~\cite{Atiyah1990,Salton:2016qpp}. Since integrating over a single $\eta_i$ with a spin-$1/2$ Wilson loop prepares the operator $X_i$, as defined in Section~\ref{FusionTensorSection}, path integration over $\eta_1 \sqcup \eta_2$ prepares $X_1\otimes X_2$ as
\begin{equation}\label{XXFactorization}
Z(\eta_1 \sqcup \eta_2) = Z(\eta_1) \otimes Z(\eta_2) \;\;\longrightarrow\;\; X_1 \otimes X_2.
\end{equation}

Acting with $X_1\otimes X_2$ on a $2$-qubit computational basis state, and expressing each factor using the fusion tensor in Eq.\ \eqref{XMatrixElements}, gives
\begin{equation}\label{XXConstruction}
X_1 \otimes X_2\,\ket{j}\ket{\ell} = N^{j+1\bmod 2}_{j,\,1}\,N^{\ell+1\bmod 2}_{\ell,\,1}\,\ket{j+1\bmod 2}\ket{\ell+1\bmod 2},
\end{equation}
which describes the tensor product of the two shift operations, one acting on each qubit. The manifold $\eta_1 \sqcup \eta_2$ is sufficient to realize $X_1\otimes X_2$ since each factor acts independently on a separate boundary tori, one for each $\eta$. 

The Pauli string $X_1\otimes X_2$, realized by path integration over $\eta_1\sqcup\eta_2$, serves as the generator of the Ising interaction gate $U(\theta)$ defined
\begin{equation}\label{GateDef}
U(\theta) = \exp\left(-i\,\frac{\theta}{2}\,X_1\otimes X_2\right), \qquad \theta \in [0,\pi].
\end{equation}
Since $(X_1\otimes X_2)^2 = \mathbb{I}$, the operator $X_1\otimes X_2$ is involutory, and the matrix exponential reduces%
\footnote{For any matrix $A$ satisfying $A^2 = \mathbb{I}$, the Euler identity $e^{-i\alpha A} = \cos(\alpha)\,\mathbb{I} - i\sin(\alpha)\,A$ offers a simplified form for the exponentiated matrix $A$.}%
 to
\begin{equation}\label{GateExpanded}
U(\theta) = \cos\left(\frac{\theta}{2}\right)\mathbb{I} - i\sin\left(\frac{\theta}{2}\right) \left(X_1\otimes X_2\right).
\end{equation}
The gate $U(\theta)$ is a Clifford unitary at three distinct values for $\theta$. For $\theta = 0$ the gate $U(\theta)$ reduces to the identity, at $\theta = \pi$ it is proportional to $X_1\otimes X_2$, and at $\theta = \pi/2$ it becomes $1/\sqrt{2}(\mathbb{I} - iX_1\otimes X_2)$, which is equivalent to $C_{\mathrm{sum}}$, introduced in Eq.\eqref{CSum}, up to local unitary action. For all $\theta\in(0,\pi/2)\cup(\pi/2,\pi)$ the gate $U(\theta)$ is non-Clifford and, as we will now show, produces non-local magic.

Consider the bipartite Hilbert space $\mathcal{H} = \mathcal{H}_A \otimes \mathcal{H}_B$, where each factor $\mathcal{H}_A$ and $\mathcal{H}_B$ is a copy of the $SU(2)_1$ torus Hilbert space $\mathcal{H}_{\mathbb{T}^2}$. Applying the operator Schmidt decomposition from Eq.\ \eqref{OpSchmidtDecomp} to Eq.\ \eqref{GateExpanded}, with $d = 2$ and $n = 2$, gives
\begin{equation}\label{SchmidtDecomp}
\frac{U(\theta)}{2} = \cos\left(\frac{\theta}{2}\right)\frac{\mathbb{I}_1}{\sqrt{2}}\otimes\frac{\mathbb{I}_2}{\sqrt{2}} - i\sin \left(\frac{\theta}{2}\right)\frac{X_1}{\sqrt{2}}\otimes\frac{X_2}{\sqrt{2}},
\end{equation}
with Schmidt eigenvalues $\lambda_1 = \cos^2(\theta/2)$ and $\lambda_2 = \sin^2(\theta/2)$. The second term of Eq.\ \eqref{SchmidtDecomp} is the $X_1 \otimes X_2$ operator realized by path integration over $\eta_1\sqcup\eta_2$, following Eq.\ \eqref{XXConstruction}, and the first is simply the identity. The parameter $\theta$ controls the relative weight between the two contributions.

Substituting eigenvalues $\lambda_1$ and $\lambda_2$ into Eq.\ \eqref{ELimReduction} gives the operator linear entropy $E_{\mathrm{lin}}$ of $U(\theta)$ as
\begin{equation}\label{ElinU}
E_{\mathrm{lin}}\bigl(U(\theta)\bigr) = 1 - \lambda_1^2 - \lambda_2^2 = \frac{1}{2}\sin^2(\theta).
\end{equation}
The value of $E_{\mathrm{lin}}\bigl(U(\theta)\bigr)$ measures the operator entanglement of the gate $U(\theta)$ itself. This quantity vanishes at $\theta = 0$ and $\theta = \pi$, when a single Schmidt term dominates, and achieves a maximum value of $1/2$ at $\theta = \pi/2$, when both Schmidt eigenvalues have equal weight.

\subsection{Non-Local Magic from Conjugated Pauli Strings}\label{subsec:nonlocal_magic}

We now address the criterion for a unitary to generate non-local magic given by Eq.\ \eqref{NonLocalMagicCriterion}. For a unitary $U$ conjugating all $n$-qubit Pauli strings $\mathcal{P}_n$, Eq.\ \eqref{NonLocalMagicCriterion} requires that $E_{\mathrm{lin}}(U^\dagger P\,U) > 0$ for at least one $P \in \mathcal{P}_n$. 

Examining $U(\theta)$ in Eq.\ \eqref{GateExpanded}, we observe that the action of $U^{\dagger}(\theta)$ will fix all Pauli strings that commute with $X_1\otimes X_2$, and rotate all anti-commuting Pauli strings by an angle $\theta$. Therefore, if $\{P,\ X_1\otimes X_2\} = 0$ then 
\begin{equation}\label{PauliRotation}
    U(\theta)^\dagger P\, U(\theta) = \cos(\theta)\,P - i\sin(\theta)\,P(X_1\otimes X_2),
\end{equation}
produces a superposition of $P$ with the operator $P(X_1\otimes X_2)$. Non-local magic is generated when $P$ is a local operator, but $P(X_1\otimes X_2)$ is a genuinely bipartite operator and therefore acts non-trivially across the $A|B$ Hilbert space partition. Accordingly, the gate $U(\theta)$ generates non-local magic, which we will now prove.

\begin{proposition}\label{prop:nonlocal}
For $\theta \in (0,\pi/2)\cup(\pi/2,\pi)$, the Ising interaction gate $U(\theta)$ defined in Eq.\ \eqref{GateDef} generates non-local magic. Specifically for $P = Z_1\otimes\mathbb{I}_2$,
\begin{equation}\label{ConjugationResult}
U(\theta)^\dagger\,(Z_1\otimes\mathbb{I}_2)\,U(\theta) = \cos(\theta)\,(Z_1\otimes\mathbb{I}_2) - \sin(\theta)\,(Y_1\otimes X_2),
\end{equation}
and the operator linear entropy of the conjugated Pauli string satisfies
\begin{equation}\label{ElinResult}
E_{\mathrm{lin}}\left(U(\theta)^\dagger\,(Z_1\otimes\mathbb{I}_2)\,U(\theta)\right) = \frac{1}{2}\sin^2(2\theta) > 0.
\end{equation}
\end{proposition}

\begin{proof}
Let $G \equiv X_1\otimes X_2$ such that $U(\theta)$ and $U^{\dagger}(\theta)$ may be written
\begin{equation}\label{Rewriting}
\begin{split}
    U(\theta) &= \cos(\theta/2)\,\mathbb{I} - i\sin(\theta/2)\,G,\\
    U(\theta)^\dagger &= \cos(\theta/2)\,\mathbb{I} + i\sin(\theta/2)\,G.\\
\end{split}
\end{equation}
Conjugating $P = Z_1\otimes\mathbb{I}_2$ by $U(\theta)$ and expanding, we have
\begin{align}\label{ConjugationExpansion}
U(\theta)^\dagger(Z_1\otimes\mathbb{I}_2)U(\theta) &= \cos^2\frac{\theta}{2}\,(Z_1\otimes\mathbb{I}_2) + i\cos\frac{\theta}{2}\sin\frac{\theta}{2}\bigl[G,\,Z_1\otimes\mathbb{I}_2\bigr] \notag\\
&\quad \quad + \sin^2\frac{\theta}{2}\,G\,(Z_1\otimes\mathbb{I}_2)\,G.
\end{align}
Employing the Pauli relations $X_1 Z_1 = iY_1$ and $Z_1 X_1 = -iY_1$, we can expand the commutator $\bigl[G,\,Z_1\otimes\mathbb{I}_2\bigr]$ as
\begin{equation}
\bigl[G,\,Z_1\otimes\mathbb{I}_2\bigr] = (X_1 Z_1 - Z_1 X_1)\otimes X_2 = 2iY_1\otimes X_2,
\end{equation}
and therefore the second term in Eq.\ \eqref{ConjugationExpansion} contributes
$-\sin(\theta)\,(Y_1\otimes X_2)$. For the final term in Eq.\ \eqref{ConjugationExpansion} we have
\begin{equation}
    G(Z_1\otimes\mathbb{I}_2)G = (X_1 Z_1 X_1)\otimes\mathbb{I}_2 = -Z_1\otimes\mathbb{I}_2.
\end{equation}
Applying the relation $\cos^2(\theta/2) - \sin^2(\theta/2) = \cos(\theta)$ then gives Eq.\ \eqref{ConjugationResult}. We can observe that the local operator $Z_1\otimes\mathbb{I}_2$ has been rotated into a superposition with the bipartite operator $Y_1\otimes X_2$, with mixing parameter $\theta$.

Establishing positivity of the operator linear entropy in Eq.\ \eqref{ElinResult}, we first write the conjugated operator%
\footnote{Conjugation by a unitary preserves unitarity. The operator $U^\dagger P U$ is unitary, therefore the operator Schmidt decomposition of Eq.\ \eqref{OpSchmidtDecomp} carries the normalization $(U^\dagger P U)/\sqrt{d^n} = (U^\dagger P U)/2$.} %
from Eq.\ \eqref{ConjugationResult} in its Schmidt form
\begin{equation}\label{SchmidtFormOp}
\frac{1}{2}\left(\cos(\theta)\,(Z_1\otimes\mathbb{I}_2) - \sin(\theta)\,(Y_1\otimes X_2)\right) = \cos(\theta)\,\frac{Z_1}{\sqrt{2}}\otimes\frac{\mathbb{I}_2}{\sqrt{2}} - \sin(\theta)\,\frac{Y_1}{\sqrt{2}}\otimes\frac{X_2}{\sqrt{2}}.
\end{equation}
Both terms in Eq.\ \eqref{SchmidtFormOp} have orthogonal Hilbert-Schmidt inner product, since $\mathrm{Tr}(Z^\dagger Y) = 0$ and $\mathrm{Tr}(\mathbb{I}^\dagger X) = 0$, yielding Schmidt eigenvalues 
\begin{equation}
\begin{split}
    \mu_1 &= \cos^2(\theta),\\
    \mu_2 &= \sin^2(\theta).
\end{split}
\end{equation}
Applying the reduced operator linear entropy form in Eq.\ \eqref{ELimReduction} gives
\begin{equation}
E_{\mathrm{lin}}(U(\theta)) = 1 - \mu_1^2 - \mu_2^2 = \frac{1}{2}\sin^2(2\theta),
\end{equation}
which is strictly positive for $\theta\in(0,\pi/2)\cup(\pi/2,\pi)$. Therefore the Ising interaction gate $U(\theta)$ generates non-local magic according to the criterion in Eq.\ \eqref{NonLocalMagicCriterion}.
\end{proof}

\subsection{Non-Stabilizing Power of Ising Gate}\label{subsec:nsp}

We now compute the average magic generated by the Ising interaction gate when acting on stabilizer states, according to the non-stabilizing power $m_p(U)$ defined in Eq.\ \eqref{NSPDef}. At two qubits there are $60$ stabilizer states, $12$ of which are eigenstates of $U(\theta)$ and acquire no magic under its action. For each of the remaining $48$ states $\ket{\psi}$, the linear stabilizer entropy of $U(\theta)\ket{\psi}$, computed using Eq.\ \eqref{LinearStabEntropy}, is
\begin{equation}
    \mathcal{M}(U(\theta)|\psi\rangle) = \sin^2(\theta)\cos^2(\theta),
\end{equation}
as verified by direct calculation of the Pauli expectation values~\cite{Varikuti:2025poy}. The average over all $60$ stabilizer states is then
\begin{equation}\label{NSPResult}
m_p\bigl(U(\theta)\bigr) = \frac{48}{60}\sin^2(\theta)\cos^2(\theta) = \frac{1}{5}\sin^2(2\theta).
\end{equation}
We highlight that $m_p\bigl(U(\theta)\bigr)$ vanishes for $\theta = 0$, $\pi/2$, and $\pi$, and reaches a maximum value of $1/5$ when $\theta = \pi/4$.

Table~\ref{tab:three_measures} collects the three operator non-stabilizer quantities computed in this section.
\begin{table}[h]
\centering
\renewcommand{\arraystretch}{1.4}
\begin{tabular}{llcc}
\hline
Quantity & Definition & Formula & Peaks at \\
\hline
$E_{\mathrm{lin}}(U(\theta))$ & Op.\ entanglement of $U$ & $\frac{1}{2}\sin^2(\theta)$ & $\theta = \pi/2$ \\
$E_{\mathrm{lin}}(U^\dagger(Z_1\otimes\mathbb{I}_2)U)$ & Op.\ entanglement of conj. Pauli & $\frac{1}{2}\sin^2(2\theta)$ & $\theta = \pi/4,\,3\pi/4$ \\
$m_p(U(\theta))$ & Non-stabilizing power & $\frac{1}{5}\sin^2(2\theta)$ & $\theta = \pi/4,\,3\pi/4$ \\
\hline
\end{tabular}
\caption{Three operator entanglement measures for the gate $U(\theta)$ in $SU(2)_1$. $E_{\mathrm{lin}}(U)$ gives the operator entanglement of $U(\theta)$ itself. $E_{\mathrm{lin}}(U^\dagger PU)$ detects if conjugation by $U$ maps a local Pauli to a bipartite operator, following Eq.\ \eqref{NonLocalMagicCriterion}. $m_p$ measures the average magic generated on stabilizer states, and shares the $\sin^2(2\theta)$ dependence with $E_{\mathrm{lin}}$.}
\label{tab:three_measures}
\end{table}
Both $E_{\mathrm{lin}}(U^\dagger(Z_1\otimes\mathbb{I}_2)U)$ and $m_p(U(\theta))$ share a $\sin^2(2\theta)$ dependence, while $E_{\mathrm{lin}}(U(\theta))$ depends on $\sin^2(\theta)$. Accordingly, $E_{\mathrm{lin}}(U(\theta))$ peaks at $\theta = \pi/2$ where the other two quantities vanish. We emphasize that the shared functional form of $E_{\mathrm{lin}}(U^\dagger PU)$ and $m_p(U(\theta))$ is not true in general, but is specific to this gate since both quantities are determined by the action of $U^{\dagger}(\theta)$ on the same anti-commuting two-qubit Pauli strings.

In this section we demonstrated that path integration in $SU(2)_1$ Chern-Simons theory can prepare the Ising interaction gate $U(\theta)$, with variable non-local magic set by a single parameter $\theta$. The origin of the $U(\theta)$ generator $X_1 \otimes X_2$ in the $SU(2)_1$ fusion structure provides an explicit connection between the algebraic data of the TQFT and the resource-theoretic properties of the resulting gate. Fundamentally $U(\theta)$ is a two-qubit construction generated from a single Pauli string. The underlying $\mathbb{Z}_2$ fusion rules that enable this construction are likewise limited, since $a \otimes b = a + b \pmod{2}$ can only determine input parity, rather than individual values. In Section~\ref{sec:su2k} we show that this obstruction prevents a topological realization of the three-qubit Toffoli gate in $SU(2)_1$, and identify the minimal level at which this gate becomes accessible.

\section{Towards Topological Toffoli Gates in $SU(2)_k$}\label{sec:su2k}

The Ising interaction gate of Section~\ref{sec:magic} demonstrates that non-local magic is accessible within $SU(2)_1$, but its construction relies on a single two-qubit Pauli string. In this section we ask whether more complex magic gates can be realized topologically, using the three-qubit Toffoli gate as an explicit target. We first show that the $\mathbb{Z}_2$ fusion structure of $SU(2)_1$ obstructs Toffoli gate construction, then identify $SU(2)_3$ as the minimal level whose fusion rules admit the required conditional logic. We describe the necessary properties of any manifold which realizes Toffoli gate construction via path integration, establish the necessity of its existence through a mapping class group density argument, and state remaining open problems.

\subsection{The $\mathbb{Z}_2$ Obstruction in $SU(2)_1$}\label{subsec:obstruction}

The Toffoli gate, denoted $\mathrm{CCNOT}_{a,b,c}$, acts on three qubits by applying $X$ to the target qubit $c$ if and only if both control qubits $a$ and $b$ are in state $\ket{1}$, stated
\begin{equation}\label{ToffoliAction}
\mathrm{CCNOT}_{a,b,c}\ket{a, b, c} = \ket{a,\, b,\, c \oplus ab}.
\end{equation}
Realizing the Toffoli gate requires distinguishing the control configuration $\ket{11}$ from configurations $\ket{00}$, $\ket{01}$, and $\ket{10}$. The AND condition necessary to implement this distinction cannot be reduced to the parity condition imposed by the $SU(2)_1$ fusion rules, as we now show.

In $SU(2)_1$ there are two representations, $\{0, 1/2\}$, and the single non-trivial fusion matrix is therefore given by
\begin{equation}\label{N12_SU21}
N_{\frac{1}{2}} = \begin{pmatrix} 
0 & 1 \\ 
1 & 0 
\end{pmatrix},
\end{equation}
where rows and columns are ordered as $\{0, 1/2\}$. Each row and column contains exactly one non-zero entry, so the fusion of any two representations produces a unique output with no branching into multiple fusion channels. Accordingly, $N_{1/2}$ reproduces the $\mathbb{Z}_2$ addition rule $a \otimes b = a + b \pmod{2}$.

Every amplitude computed by path integration in $SU(2)_1$ is built from contractions of $N_{1/2}$, following Eq.\ \eqref{StateAmplitudes} and the surgery rules of~\cite{Salton:2016qpp}, and therefore depends on the input representation labels only through their sum modulo $2$. In particular, the control configurations $\ket{1/2, 1/2}$ and $\ket{0, 0}$ both fuse to the trivial channel
\begin{equation}\label{LevelOneRule}
\frac{1}{2} + \frac{1}{2} = 0 \pmod{2}, \qquad 0 + 0 = 0 \pmod{2},
\end{equation}
and are therefore indistinguishable to any path integral in $SU(2)_1$. Since the Toffoli gate requires distinguishing these two configurations, the AND condition is inaccessible at this level.

\subsection{Fusion Rules of $SU(2)_3$}\label{subsec:su23_fusion}

We next show that $SU(2)_3$ is the minimal level required to realize the Toffoli gate, with fusion rules rich enough to support the conditional logic necessary to implement this operation. Recall that representations of $SU(2)_k$ Chern-Simons theory are labeled by spins $j \in \{0,\ 1/2,\ 1,\ \ldots,\ k/2\}$, and their fusion rules follow a truncated Clebsch-Gordan series~\cite{Witten1989,MooreSeiberg1989}
\begin{equation}\label{TruncatedCG}
j_1 \otimes j_2 = \bigoplus_{j_3 = |j_1-j_2|}^{\min(j_1+j_2,\, k-j_1-j_2)} j_3,
\end{equation}
where the sum occurs in integer steps and the upper bound is determined by the level-$k$ truncation. 

For $k = 3$ the integrable representations are $j \in \{0,\ 1/2,\ 1,\ 3/2\}$, and the associated torus Hilbert space $\mathcal{H}_{\mathbb{T}^2} \cong \mathbb{C}^4$, as described in Section~\ref{ModularSection}, has dimension $d = k + 1 = 4$. For each representation $a$ we associate a fusion matrix $N_a$, with matrix elements $N^c_{ab}$ that encode the number of ways $a$ and $b$ fuse to produce $c$. As before, the rows and columns of each $N^c_{ab}$ are ordered as $\{0, \frac{1}{2}, 1, \frac{3}{2}\}$, and the four matrices are given by
\begin{equation}\label{FusionMatrices_SU23}
\begin{aligned}
N_0 &= \begin{pmatrix} 
1&0&0&0\\
0&1&0&0\\
0&0&1&0\\
0&0&0&1 
\end{pmatrix}, 
\qquad N_{\frac{1}{2}} = \begin{pmatrix} 
0&1&0&0\\
1&0&1&0\\
0&1&0&1\\
0&0&1&0 
\end{pmatrix}, \\[8pt]
N_1 &= \begin{pmatrix} 
0&0&1&0\\
0&1&0&1\\
1&0&1&0\\
0&1&0&1 
\end{pmatrix}, 
\qquad N_{\frac{3}{2}} = \begin{pmatrix} 
0&0&0&1\\
0&0&1&0\\
0&1&0&1\\
1&0&1&0 
\end{pmatrix}.
\end{aligned}
\end{equation}
Matrix $N_0$ is the identity and reflects trivial fusion with the vacuum. The first distinction of $SU(2)_3$ from $SU(2)_1$ appears in $N_{1/2}$, where the second row, corresponding to $b = 1/2$, contains non-zero entries in columns $c = 0$ and $c = 1$. Consequently, the fusion of two spin-$1/2$ representations branches into two channels,
\begin{equation}\label{BranchingRule}
\frac{1}{2} \otimes \frac{1}{2} = 0 \oplus 1,
\end{equation}
each corresponding to a distinct topological sector of the path integral. 

The branching in Eq.\ \eqref{BranchingRule} is absent at level $k = 1$, as shown in Eq.\ \eqref{LevelOneRule}. Branching does appear at $k = 2$, where the fusion rule gives $1/2 \otimes 1/2 = 0 \oplus 1$ with three available representations $\{0,\ 1/2,\ 1\}$. However, as we discuss further in Section~\ref{subsec:existence}, Toffoli gate construction additionally requires~\cite{Freedman2003} that the mapping class group representation be dense in the projective unitary group of $\mathcal{H}_{\Sigma_g}$. This property is satisfied%
\footnote{At $k = 4$ the image of the MCG in the projective unitary group is finite rather than dense, so arbitrary unitaries cannot be approximated to arbitrary precision. Levels $k = 1,\ 2,\ 4$ are exceptional cases where the density condition fails~\cite{Freedman2003}.} %
for $k \geq 3$, with $k \neq 4$, but fails at $k = 2$. At $k = 3$ both the branching structure and the density property are present, making it the minimal level at which a topological realization of the Toffoli operator can exist. As with the $SU(2)_1$ case introduced in Section~\ref{FusionTensorSection}, each fusion matrix $N_a$ is realized by path integration over $\eta$ with a Wilson loop in the representation $a$, preparing a tripartite state with amplitudes proportional to $N^{j_3}_{j_1, a}$ following Eq.\ \eqref{EtaState}.

\subsection{Logical Encoding and Structure of the Toffoli Manifold}\label{subsec:encoding}

The $SU(2)_3$ torus Hilbert space is four-dimensional, therefore realizing a qubit gate in this Hilbert space requires specifying how the logical qubit is encoded. One natural choice for this encoding is to use the same two representations which span $\mathcal{H}_{\mathbb{T}^2}$ in $SU(2)_1$, that is
\begin{equation}\label{LogicalEncoding}
\ket{0}_L \equiv \ket{j=0}, \qquad \ket{1}_L \equiv \ket{j=\frac{1}{2}},
\end{equation}
defining the two-dimensional logical subspace $\mathcal{H}_L \equiv \mathrm{span}\{\ket{0}, \ket{1/2}\} \subset \mathcal{H}_{\mathbb{T}^2}$. The orthogonal complement subspace $\mathcal{H}_L^\perp = \mathrm{span}\{\ket{1},\ket{3/2}\}$ then serves as the leakage subspace, and any fusion process that takes an amplitude into representations $j = 1$ or $j = 3/2$ moves the system outside this logical encoding.

The restriction of $N_{1/2}$ to the logical subspace $\mathcal{H}_L$ in Eq.\ \eqref{LogicalEncoding} is
\begin{equation}\label{N12_logical}
N_{\frac{1}{2}}\big|_{\mathcal{H}_L} = 
\begin{pmatrix} 
0 & 1 \\ 
1 & 0 
\end{pmatrix}_L = X_L,
\end{equation}
which defines the logical Pauli $X$ operator. Inside $\mathcal{H}_L$, the $\{0,\ 1/2\}$ sector of $N_{1/2}$ in $SU(2)_3$ coincides with the full fusion matrix $N_{1/2}$ of $SU(2)_1$, from Eq.\ \eqref{N12_SU21}. Accordingly, the topological construction of the logical Pauli $X$ operator follows directly from Section~\ref{FusionTensorSection}.

When two logical qubits are fused in $SU(2)_3$, the encoding of Eq.\ \eqref{LogicalEncoding} identifies
\begin{equation}\label{LogicalInput}
    \ket{1}_L\otimes\ket{1}_L = \ket{\frac{1}{2}}\otimes\ket{\frac{1}{2}}.
\end{equation}
The branching rule in Eq.\ \eqref{BranchingRule} then implies that the input in Eq.\ \eqref{LogicalInput} fuses into a superposition of $\ket{0}$, which lies in $\mathcal{H}_L$, and $\ket{1}$, which lies in $\mathcal{H}_L^\perp$. No other logical pair produces this leakage since $0\otimes 0 = 0$, $0\otimes 1/2 = 1/2$, and $1/2\otimes 0 = 1/2$ remain entirely in $\mathcal{H}_L$. The selective leakage of the $\ket{11}_L$ configuration is a topological signature of the AND condition, since the spin-$1$ intermediate channel is accessible only when both control qubits are spin-$1/2$. Our goal is then to construct a manifold whose path integral utilizes this distinction to implement the conditional $X$ operation on the target qubit.

We realize the Ising interaction gate in Section~\ref{sec:magic} by path integration over the disjoint union $\eta_1\sqcup\eta_2$, which factors into independent operations on each boundary torus. The Toffoli gate cannot be realized this way, since its action on the target qubit depends on the joint state of both control qubits, rather than on each control qubit independently. The required manifold, which we denote $\mathcal{M}_T$, must therefore be connected such that the bulk topology couples all three qubit boundaries together. Following the Choi-Jamio\l{}kowski convention of~\cite{Salton:2016qpp}, $\mathcal{M}_T$ will have six torus boundaries, two for each logical qubit with one input and one output. Path integration over $\mathcal{M}_T$ must implement the map
\begin{equation}\label{ToffoliAmplitudes}
\bra{a',b',c'}\mathcal{M}_T\ket{a,b,c}_L = \delta_{a',a}\,\delta_{b',b}\,\delta_{c',\, c\oplus ab},
\end{equation}
where $a, b, c \in \{0, 1\}_L$ serve as logical qubit labels and $\oplus$ indicates addition mod $2$. Using the $\{0,\ 1/2\}$ representation labels of the Eq.\ \eqref{LogicalEncoding} encoding, the amplitudes in Eq.\ \eqref{ToffoliAmplitudes} require that the path integral acts as $\mathbb{I}$ on the target boundary for control configurations $(a,b) \neq (1/2, 1/2)$, and implements $X_L$ when $(a,b) = (1/2, 1/2)$. This conditional structure cannot be achieved by any disjoint union of handlebodies, since disjoint unions can only implement the tensor product of independent operators. 

Following the above argument, the manifold $\mathcal{M}_T$ must therefore be connected and its topology must encode the necessary correlation between control and target boundaries. The $SU(2)_3$ fusion data implements the requisite conditional action via a two-step process. First, the two Wilson loops corresponding to the control qubits, each carrying spin-$1/2$, join at a trivalent junction in the interior of $\mathcal{M}_T$. The branching rule of Eq.\ \eqref{BranchingRule} produces contributions in both the spin-$0$ and spin-$1$ representations when the control configuration is $\ket{11}_L$, while all other control configurations fuse entirely within $\mathcal{H}_L$. The spin-$1$ contribution, which is internal to $\mathcal{M}_T$, then couples to the target boundary through the fusion matrix $N_1$. The fusion matrix entry
\begin{equation}\label{N1_entry}
(N_1)_{\frac{1}{2},\frac{1}{2}} = 1
\end{equation}
ensures that the fusion $1/2 \otimes 1$ contains a spin-$1/2$ channel, enabling amplitude to return to the logical subspace while implementing a conditional flip on the target qubit.

We summarize below the required properties of a manifold $\mathcal{M}_T$ that topologically realizes the Toffoli gate.
\begin{enumerate}
\item $\mathcal{M}_T$ is a connected, oriented $3$-manifold with six torus boundaries.
\item The interior of $\mathcal{M}_T$ contains two Wilson loops, in the spin-$1/2$ representation of $SU(2)_3$, one for each control qubit.
\item These Wilson loops meet with fusion amplitude $N^{j_3}_{1/2,\ 1/2}$, and produce contributions in the spin-$0$ and spin-$1$ representations. Spin-$1$ is an intermediate label, internal to $\mathcal{M}_T$, and is not associated to any boundary.
\item The spin-$1$ representation couples to the target boundary through a second fusion  $N^{j_3}_{1/2,\ 1}$, implementing a conditional flip on the target qubit.
\item Path integration cancels all leakage into $\mathcal{H}_L^\perp$, such that the amplitudes in Eq.\ \eqref{ToffoliAmplitudes} are realized fully within the logical subspace.
\end{enumerate}
The fifth condition above warrants further discussion. The branching $1/2 \otimes 1/2 = 0 \oplus 1$ directs amplitude into both spin-$0$ and spin-$1$ channels, and the relative phases between each channel are determined by the $SU(2)_3$ modular data. Specifically, relative phases depend on the $S$-matrix and topological spins of the theory~\cite{Witten1989}, and enter the path integral through Dehn surgery on $\mathcal{M}_T$. In order for the Toffoli gate to act correctly within $\mathcal{H}_L$, we require that the spin-$0$ channel produces the identity action on the target qubit, while the spin-$1$ channel initiates a bit flip. Moreover, all leakage components throughout the operation must cancel by destructive interference. In Section~\ref{subsec:existence}, we establish the existence of $\mathcal{M}_T$ and identify the details of this leakage cancellation as an open problem.

A natural ansatz for $\mathcal{M}_T$ is obtained by gluing two handlebodies along an intermediate torus boundary. The first handlebody $\eta_{12}$ contains the two control Wilson loops, and fuses them to produce the spin-$0$ and spin-$1$ intermediate channels via Eq.\ \eqref{BranchingRule}. The output boundary of $\eta_{12}$, which carries the intermediate representation, is then glued to the input boundary of a second handlebody $\eta_{23}$, which fuses the intermediate channel with the target qubit. Schematically, we describe $\mathcal{M}_T$ as
\begin{equation}\label{ToffoliAnsatz}
\mathcal{M}_T \sim \eta_{12} \cup_{\mathbb{T}^2} \eta_{23},
\end{equation}
where the gluing of handlebodies is performed along the shared torus boundary that carries the intermediate representation. This ansatz encodes the two-step fusion described above, with $\eta_{12}$ implementing the control action and $\eta_{23}$ the conditional bit flip on the target qubit. However, determining whether $\mathcal{M}_T$ in Eq.\ \eqref{ToffoliAnsatz} satisfies the leakage cancellation from Condition 5 requires computing the relative phases between the spin-$0$ and spin-$1$ channels. As discussed above, these relative phases depend on the surgery presentation of $\mathcal{M}_T$ in $S^3$ and modular data of $SU(2)_3$.

\subsection{Existence via Density and Open Problems}\label{subsec:existence}

In this section we establish that the Toffoli gate can be approximated to arbitrary precision by path integration in $SU(2)_k$, for suitable choice of $k$, and state the aspects of this problem which remain open. The key idea follows a density result from Kitaev et al.~\cite{Freedman2003} which states: For $SU(2)_k$ with $k \geq 3$ and $k \neq 4$, the representation of the mapping class group for any surface $\Sigma_g$ of genus $g \geq 2$ is dense in the projective unitary group of $\mathcal{H}_{\Sigma_g}$. This density result, combined with the existence of the path-integral isometry~\cite{Salton:2016qpp}
\begin{equation}\label{EmbeddingIsometry}
    V: \mathcal{H}_{\mathbb{T}^2}^{\otimes 2} \to \mathcal{H}_{\Sigma_2},
\end{equation}
implies that any unitary on $\mathcal{H}_{\mathbb{T}^2}^{\otimes 2}$ can be approximated by path integration up to overall rescaling. The case of $k = 4$ is excluded because the mapping class group at this level is finite, and thereby not dense~\cite{Freedman2003}.

For a three-qubit Toffoli gate, we require a unitary on $\mathcal{H}_{\mathbb{T}^2}^{\otimes 3}$ which embeds into $\mathcal{H}_{\Sigma_g}$ for $g \geq 3$ through an isometry analogous to Eq.\ \eqref{EmbeddingIsometry}. The density theorem of \cite{Freedman2003} applies for any genus $g \geq 2$, so the argument extends faithfully.

\begin{proposition}[Existence of Toffoli in $SU(2)_k$]\label{prop:toffoli_existence}
For an $SU(2)_k$ Chern-Simons theory with $k \geq 3$ and $k \neq 4$, with the logical encoding given by Eq.\ \eqref{LogicalEncoding}, the Toffoli gate acting on three logical qubits can be approximated to arbitrary precision by path integration over a connected, oriented $3$-manifold with Wilson loop insertions.
\end{proposition}

\begin{proof}
The Toffoli gate acts as a unitary on $\mathcal{H}_L^{\otimes 3} \subset \mathcal{H}_{\mathbb{T}^2}^{\otimes 3}$, which can be extended to a unitary on $\mathcal{H}_{\mathbb{T}^2}^{\otimes 3}$ by acting as the identity on the logical complement subspace $(\mathcal{H}_L^\perp)^{\otimes 3}$. The isometry 
\begin{equation}
    V: \mathcal{H}_{\mathbb{T}^2}^{\otimes 3} \to \mathcal{H}_{\Sigma_3},
\end{equation}
constructed by path integration over a genus-3 handlebody with three solid tori removed from its interior, which generalizes the genus-2 construction of \cite{Salton:2016qpp}, embeds $\mathcal{H}_{\mathbb{T}^2}^{\otimes 3}$ into $\mathcal{H}_{\Sigma_3}$. Following the density theorem~\cite{Freedman2003}, the mapping class group of $\Sigma_3$ admits a dense representation in the projective unitary group of $\mathcal{H}_{\Sigma_3}$ for $k \geq 3$, $k \neq 4$. Any unitary $U$ on $\mathcal{H}_{\mathbb{T}^2}^{\otimes 3}$ can therefore be approximated~\cite{Salton:2016qpp} as
\begin{equation}
    U \approx V^\dagger W_f V, 
\end{equation}
where $W_f$ is a unitary on $\mathcal{H}_{\Sigma_3}$ induced by the orientation-preserving diffeomorphism $f$ of $\Sigma_3$. The operator $V^\dagger W_f V$ is realized~\cite{Salton:2016qpp} by path integration over the manifold $-\mathcal{V} \cup_f \mathcal{V}$. Restriction to the logical subspace $\mathcal{H}_L^{\otimes 3}$ then provides an approximation to the Toffoli gate, up to arbitrary precision.
\end{proof}

Proposition~\ref{prop:toffoli_existence} ensures that $\mathcal{M}_T$ exists, but does not offer its explicit construction. In order to construct $\mathcal{M}_T$ directly, the following two open problems must be resolved:

\textbf{Open Problem 1: Explicit surgery presentation.} One must construct the surgery link $\mathcal{L} \subset S^3$, and framing data, such that Dehn surgery on $\mathcal{L}$ results in a manifold whose path integral computes the Toffoli amplitudes in Eq.\ \eqref{ToffoliAmplitudes}. A natural approach is to decompose the Toffoli gate into a product of Dehn twists in the mapping class group of $\Sigma_3$, then reconstruct $\mathcal{M}_T$ using the surgery formula~\cite{Salton:2016qpp}
\begin{equation}\label{SurgeryFormula}
\langle W(L,R)\rangle_{\mathcal{M_T}} = \frac{\langle W(L,R)\tilde{W}(\mathcal{L})\rangle_{S^3}}{\langle\tilde{W}(\mathcal{L})\rangle_{S^3}},
\end{equation}
where $\tilde{W}(\mathcal{L})$ is a product of Wilson loop operators through curves $\{C_i\}$ defined 
\begin{equation}
    \tilde{W}(\mathcal{L}) = \prod_i \sum_\ell S_{\ell 0}\, W(\mathcal{C}_i, R_\ell).
\end{equation}

\textbf{Open Problem 2: Leakage cancellation.} One must verify that the leakage into the logical complement subspace $\mathcal{H}_L^\perp$ is canceled out in the path integral over $\mathcal{M}_T$. Confirming this feature requires computing relative phases between the spin-$0$ and spin-$1$ channels using the modular data of $SU(2)_3$. The $S$-matrix of $SU(2)_k$ is~\cite{Witten1989}
\begin{equation}\label{SMatrix_SU2k}
S_{jj'} = \sqrt{\frac{2}{k+2}}\sin\left(\frac{\pi(2j+1)(2j'+1)}{k+2}\right),
\end{equation}
and the topological spins are $\theta_j = e^{2\pi i h_j}$ with conformal dimensions 
\begin{equation}
    h_j = \frac{j(j+1)}{(k+2)}.
\end{equation}
For $k = 3$, the data determine the relative amplitudes of the spin-$0$ and spin-$1$ channels in the path integral, and whether the Toffoli amplitudes in Eq.\ \eqref{ToffoliAmplitudes} are exactly or only approximately reproduced in $\mathcal{H}_L^{\otimes 3}$.

The fusion structure described in Section~\ref{subsec:encoding} offers a guide for both open problems. The junction where the two control Wilson loops meet corresponds to a surgery curve that links both control-qubit cores with an intermediate loop. The coupling of this junction to the target qubit is another curve that links the intermediate loop with the target-qubit core. A solution to open problems 1 and 2 would determine the explicit linking matrix and surgery presentation which realize this fusion structure.

The $\mathbb{Z}_2$ algebra of $SU(2)_1$ limits which gates are topologically accessible to those conditioned on a simple parity check. The more complicated Toffoli gate requires the AND conditional, which is topologically accessible beginning at $SU(2)_3$ through the fusion rule $1/2 \otimes 1/2 = 0 \oplus 1$ in Eq.\ \eqref{BranchingRule}. While the representation density of the mapping class group in the projective unitary group of the boundary Hilbert space (Kitaev et al.~\cite{Freedman2003}), coupled with the unitary approximation via manifold to boundary isometry (Salton et al.~\cite{Salton:2016qpp}), guarantees that the Toffoli gate can be approximated by path integration in $SU(2)_3$, explicit construction of the manifold $\mathcal{M}_T$ on which the path integral is performed, and subsequent verification of leakage cancellation remain open questions. In the next section we leave the Chern-Simons framework, and instead explore Dijkgraaf-Witten theory where an exact realization of the non-Clifford T gate can be achieved topologically without approximation.

\section{The T Gate in Dijkgraaf-Witten Theory}\label{sec:DW}

In Sections~\ref{sec:magic} and~\ref{sec:su2k} we use Chern-Simons theory to topologically prepare non-Clifford gates either continuously, as with the Ising interaction gate $U(\theta)$, or approximately, as with the Toffoli gate using the mapping class group density theorem. In this section we move beyond Chern-Simons theory to demonstrate how the T gate, which serves as the canonical single-qubit non-Clifford operation, can be realized by exact path integration in Dijkgraaf-Witten theory with the finite gauge group $\mathbb{Z}_4$. We explain how the non-Clifford properties of the T gate originate from the 3-cocycle data of the theory, in contrast with our Chern-Simons constructions.

\subsection{Dijkgraaf-Witten Theory with Finite Gauge Group}\label{subsec:DW_intro}

Dijkgraaf-Witten (DW) theory~\cite{Dijkgraaf:1990} is a topological quantum field theory defined by a finite gauge group $G$ and cohomology class $[\omega] \in H^3(G, U(1))$, as defined in Section~\ref{subsec:cohomology}. Since $G$ is finite, the DW path integral reduces to a finite sum over gauge field configurations, yielding an exactly solvable theory. The gauge field configuration on a closed, oriented $3$-manifold $\mathcal{M}$ is specified by a group homomorphism
\begin{equation}\label{FlatBundle}
\phi\,\colon \pi_1(\mathcal{M}) \to G,
\end{equation}
where $\pi_1(\mathcal{M})$ is the fundamental group%
\footnote{Recall the fundamental group of a manifold denotes the group of homotopy classes of loops based at a fixed point, with group operation defined as loop concatenation.} %
of $\mathcal{M}$. Eq.\ \eqref{FlatBundle} assigns a group element to each non-contractible loop in $\mathcal{M}$, compatible with loop composition. The map $\phi$ serves an analog role to specifying a flat connection in Chern-Simons theory. Since $G$ is finite, there are finitely many such homomorphisms.

Each $\phi$ determines a $U(1)$-valued topological action $\omega(\mathcal{M}, \phi)$ by the following construction. The 3-cocycle $\omega$ assigns $U(1)$ phases to triples of elements $g \in G$ (see Section~\ref{subsec:cohomology}). The homomorphism $\phi$ maps loops in $\mathcal{M}$ to elements in $G$. Composing $\omega$ with $\phi$ then builds a 3-cocycle on $\mathcal{M}$, rather than on $G$. The composed cocycle, which we denote $f_\phi^*[\omega]$, is evaluated on $\mathcal{M}$ to produce a single $U(1)$ phase
\begin{equation}\label{DWAction}
\omega(\mathcal{M}, \phi) = \langle f_\phi^*[\omega],\, [\mathcal{M}] \rangle \in U(1),
\end{equation}
which plays the role of $e^{iS_{CS}}$ in the Chern-Simons path integral of Eq.\ \eqref{PartitionFunction}. The DW partition function $Z(\mathcal{M}, \omega)$ is then expressed as the sum of phases over all gauge configurations
\begin{equation}\label{DWPartitionFunction}
Z(\mathcal{M}, \omega) = \frac{1}{|G|}\sum_{\phi\,\colon\,\pi_1(\mathcal{M})\to G} \omega(\mathcal{M}, \phi).
\end{equation}
The prefactor $|G|^{-1}$ accounts for gauge equivalences. Since the sum in Eq.\ \eqref{DWPartitionFunction} is finite, $Z(\mathcal{M}, \omega)$ can be computed exactly and is a topological invariant of $\mathcal{M}$.

When $\mathcal{M}$ has a torus boundary $\partial\mathcal{M} = \mathbb{T}^2$, path integration over $\mathcal{M}$ defines a state in the Hilbert space $\mathcal{H}_{\mathbb{T}^2}$, directly analogous to the Chern-Simons construction of Section~\ref{TopologicalStatePrep}. Basis states $\ket{j}$ are labeled by irreducible representations of $G$. For a finite group, the Hilbert space dimension counts the number of conjugacy classes of $G$, giving $\dim\mathcal{H}_{\mathbb{T}^2} = |G|$ when $G$ is abelian%
\footnote{For an abelian group, each element forms its own conjugacy class.}%
. As with the Chern-Simons case described in Section~\ref{ModularSection}, the mapping class group $\mathrm{SL}(2,\mathbb{Z})$ acts on $\mathcal{H}_{\mathbb{T}^2}$ by the modular $S$ and $T$ matrices, which are now determined by cocycle $\omega$ instead of level $k$.

\subsection{The Generating Cocycle for $\mathbb{Z}_4$ and the T Gate}\label{subsec:cocycle}

The cohomology group $H^3(\mathbb{Z}_N, U(1)) \cong \mathbb{Z}_N$, given by Eq.\ \eqref{H3ZN} of Section~\ref{subsec:cohomology}, classifies $N$ inequivalent DW theories for gauge group $G = \mathbb{Z}_N$, each labeled by the cocycle $\omega_p = \omega_1^p$ for $p \in \{0, 1, \ldots, N-1\}$, as in Eq.\ \eqref{OmegaPowers}. An explicit representative of the generating cocycle $\omega_1$ can be written~\cite{Dijkgraaf:1990} in terms of the $\mathbb{Z}_N$ representatives $\hat{a},\ \hat{b},\ \hat{c} \in \{0, 1, \ldots, N-1\}$, of elements $a,\ b,\ c \in \mathbb{Z}_N$, as
\begin{equation}\label{GeneratingCocycle}
\omega_p(a, b, c) = \exp\left(\frac{2\pi i p}{N^2}\,\hat{a}\left(\hat{b} + \hat{c} - \widehat{b+c}\right)\right),
\end{equation}
where $\widehat{b+c}$ denotes the representative of $b + c$ in $\{0, \ldots, N-1\}$, and $(\hat{b} + \hat{c} - \widehat{b+c})$ is either $0$ or $N$ depending on whether the addition overflows. The cocycle $\omega_p(a, b, c)$ satisfies the condition $\delta\omega_p = 1$, from Eq.\ \eqref{CocycleCondition}, by construction.

We now consider the group $G = \mathbb{Z}_4$ with $p = 1$. The torus Hilbert space $\mathcal{H}_{\mathbb{T}^2}$ is four-dimensional, with basis states $\{\ket{0}, \ket{1}, \ket{2}, \ket{3}\}$ labeled by elements of $G$. For abelian DW theory with gauge group $\mathbb{Z}_N$ and cocycle $\omega_p$, the modular $S$ and $T$ matrices take the form~\cite{Dijkgraaf:1990,Propitius:1995}
\begin{equation}\label{DWmodularGeneral}
S_{jk} = \frac{1}{\sqrt{N}}\exp\left(\frac{2\pi i jk}{N}\right), \qquad T_{jj} = \exp\left(\frac{\pi i p j^2}{N}\right),
\end{equation}
where $j, k \in \{0, 1, \ldots, N-1\}$. The $S$-matrix in Eq.\ \eqref{DWmodularGeneral} is independent of $p$ and implements the discrete Fourier transform, analogous to Eq.\ \eqref{ModularS} for the Chern-Simons case. The $T$-matrix depends on the cocycle parameter $p$, which multiplies $j^2$ in the exponent. For the case $N = 4$ Eq.\ \eqref{DWmodularGeneral} becomes
\begin{equation}\label{DWSmatrix}
S_{jk} = \frac{1}{2}\exp\left(\frac{2\pi i jk}{4}\right), \qquad j, k \in \mathbb{Z}_4.
\end{equation}
Similarly, for $p = 1$ the $T$-matrix
\begin{equation}\label{DWTmatrix}
T_{jj} = \exp\left(\frac{\pi i j^2}{4}\right), \qquad j \in \mathbb{Z}_4,
\end{equation}
is diagonal with specific components given by
\begin{equation}\label{DWTentries}
T_{00} = 1, \qquad T_{11} = e^{i\pi/4}, \qquad T_{22} = e^{i\pi} = -1, \qquad T_{33} = e^{9i\pi/4} = e^{i\pi/4}.
\end{equation}
We apply the logical encoding $\ket{0}_L = \ket{j=0}$ and $\ket{1}_L = \ket{j=1}$, as in Section~\ref{subsec:encoding}, in the two-dimensional subspace $\mathcal{H}_L = \mathrm{span}\{\ket{0}, \ket{1}\} \subset \mathcal{H}_{\mathbb{T}^2}$. The restriction of the modular $T$-matrix to $\mathcal{H}_L$ becomes
\begin{equation}\label{TgateFromDW}
T\big|_{\mathcal{H}_L} = \begin{pmatrix} 1 & 0 \\ 0 & e^{i\pi/4} \end{pmatrix}_L = T_{\mathrm{gate}},
\end{equation}
which is exactly the T gate without approximation. The T gate lives in the third level of the Clifford hierarchy, and is the canonical resource used for promoting Clifford operations to a universal gate set~\cite{Bravyi2004}.

\subsection{Path Integral Realization of T Gate}\label{subsec:DW_pathintegral}

The T gate in Eq.\ \eqref{TgateFromDW} admits a direct path integral interpretation. Analogous to the Chern-Simons case in Section~\ref{ModularSection}, the modular $T$ transformation is implemented geometrically by a Dehn twist on the boundary torus. In DW theory, path integration over a solid torus $\mathcal{T}$, with a Dehn-twisted boundary identification, computes the amplitude
\begin{equation}\label{DWTpathintegral}
\langle j' | \mathcal{T}_{\mathrm{twist}} | j \rangle = T_{jj}\,\delta_{j,j'} = \exp\left(\frac{\pi i j^2}{4}\right)\delta_{j,j'}.
\end{equation}
We compute the amplitude in Eq.\ \eqref{DWTpathintegral} using the DW partition function of Eq.\ \eqref{DWPartitionFunction}, which sums over $\phi : \pi_1(\mathcal{T}) \to \mathbb{Z}_4$. The boundary state $\ket{j}$ specifies that the element assigned to the non-contractible cycle of the boundary torus is $j \in \mathbb{Z}_4$. For a solid torus $\mathcal{T} \cong S^1 \times D^2$, the fundamental group is $\pi_1(\mathcal{T}) \cong \mathbb{Z}$. Fixing $j$ therefore uniquely determines the $\phi$, leaving one configuration in the sum of Eq.\ \eqref{DWPartitionFunction}. The path integral then reduces to evaluating the DW action from cocycle $\omega_1$ on this configuration, which gives the phase
\begin{equation}
\omega(\mathcal{T}, \phi) = \exp\left(\frac{\pi i j^2}{4}\right).
\end{equation}
This phase matches $T_{jj}$ in Eq.\ \eqref{DWTmatrix}, demonstrating that path integration over the solid torus, with Dehn twist, implements the modular $T$ operator. Restricting to the logical subspace $j, j' \in \{0,1\}$ then exactly reproduces the matrix elements of the T gate, as in Eq.\ \eqref{TgateFromDW}.

We highlight the fact that the path integral representation of the T gate contains a single term. In contrast with the Ising interaction gate of Section~\ref{sec:magic}, which involved a superposition of two topological sectors, i.e. the identity and the $\eta_1 \sqcup \eta_2$ contribution, the DW realization of the T gate requires no superposition or additional parameter tuning. The gate is realized by a single Dehn twist, with the non-Clifford phase arising entirely from the structure of 3-cocycle $\omega_1$.

\subsection{Comparison of Dijkgraaf-Witten and Chern-Simons Constructions}\label{subsec:DW_vs_CS}

We now compare the Dijkgraaf-Witten realization of the magic-generating T gate, with the Chern-Simons constructions of the non-local magic generating Ising and Toffoli gates in Sections~\ref{sec:magic} and~\ref{sec:su2k}. In both frameworks, gates are prepared by path integration over $3$-manifolds with torus boundaries, where the modular $S$ and $T$ matrices serve as primary building blocks. The solid torus $\mathcal{T}$, with a Dehn twist on the boundary, implements the modular $T$ transformation in both cases. However, importantly, the level of Clifford hierarchy accessed by the modular $T$ operation is different between the two theories.

In $SU(2)_1$ Chern-Simons theory, the modular $T$ matrix in Eq.\ \eqref{ModularT_SU21} acts as a diagonal phase gate at the second level of the Clifford hierarchy, gates which normalize the Pauli group. When combined with a Pauli $Z$ operation, as in Eq.\ \eqref{PhaseGate}, the modular $T$ transformation is used to construct the Clifford phase gate $P$. In this theory, therefore, the Dehn twist generates no magic. 

In DW theory with $G = \mathbb{Z}_4$ and cocycle $\omega_1$, the same geometric transformation produces a diagonal gate at the third level of the Clifford hierarchy, as in Eq.\ \eqref{TgateFromDW}. In this setting, the Dehn twist generates magic. This shift up the Clifford hierarchy reflects the different algebraic data that control the topological phases in the two theories. In Chern-Simons theory, phases in the modular $T$ matrix come from conformal dimensions $h_j = j(j+1)/(k+2)$, which yield quadratic phases in $j$. In contrast, in DW theory with cocycle $\omega_1$, the phases $\exp(\pi i j^2/4)$ also depend quadratically on $j$, but are multiplied by a different coefficient. The factor $\pi/4$, rather than a multiple of $\pi/2$, places the gate constructed by path integration in the third level of the Clifford hierarchy, rather than the second.

Another important distinction concerns the nature of the gauge group. Chern-Simons theory with a compact Lie group produces a continuously parametrized family of states and gates, exemplified by the single-parameter Ising gate $U(\theta)$ in Section~\ref{sec:magic}. DW theory with a finite gauge group produces a discrete set of gates, determined by the finite number of cocycle classes $[\omega] \in H^3(G, U(1))$. The continuous interpolation between Clifford and maximal magic gates provided by $U(\theta)$ has no analogue in the DW setting, where the allowed gates are fixed once the group $G$ and cocycle class $[\omega]$ are specified.

Finally, we emphasize that the logical encoding $\ket{0}_L = \ket{j=0}$, $\ket{1}_L = \ket{j=1}$ used to construct the T gate in the DW theory (Section~\ref{subsec:cocycle}) parallels the logical encoding $\ket{0}_L = \ket{j=0}$, $\ket{1}_L = \ket{j=1/2}$ used for the Toffoli gate construction in $SU(2)_3$ (Section~\ref{subsec:encoding}). In both cases the logical qubit is embedded in the lowest two representations of the theory. It remains an open question whether there exists a unified topological framework that can realize construction of the Ising gate, Toffoli gate, and T gate.

In this section we demonstrated how the T gate can be realized exactly by path integration over a solid torus with a single Dehn twist on its boundary in Dijkgraaf-Witten theory with $G=\mathbb{Z}_4$ and generating cocycle $\omega_1$. We explained how the resulting non-Clifford phase originates directly from the underlying 3-cocycle data. Combining these results with the Chern-Simons constructions of Sections~\ref{sec:magic} and~\ref{sec:su2k}, we established that topological path integrals produce non-Clifford gates across multiple levels of the Clifford hierarchy and within distinct classes of topological quantum field theories. We now conclude the paper with a discussion of result implications and directions for future research.

\section{Discussion}\label{sec:discussion}

In this paper we construct non-Clifford gates as path integrals in topological quantum field theories and characterize their magic content. In $SU(2)_1$ Chern-Simons theory, the two-qubit Ising interaction gate $U(\theta) = \exp(-i\theta/2 (X_1 \otimes X_2))$ is built by exponentiating the Pauli string generator $X_1 \otimes X_2$, which is prepared by path integration over the disjoint union of three-boundary handlebodies. We establish that this gate produces non-local magic for all values of $\theta$ away from a select few Clifford points, and compute its non-stabilizing power to be $m_p = 1/5(\sin^2(2\theta))$. We further show that construction of the $3$-qubit Toffoli gate in $SU(2)_1$ is obstructed by the $\mathbb{Z}_2$ fusion rules, which are only able to distinguish parity and cannot implement the AND conditional required by the Toffoli gate. Extending to $SU(2)_3$, the branching fusion rule $1/2 \otimes 1/2 = 0 \oplus 1$ resolves this obstruction, and enables us to establish the existence of a connected manifold whose path integral approximates the Toffoli gate to arbitrary precision. 

We then move beyond a Chern-Simons framework and demonstrate that the T gate, the canonical single-qubit magic-generating gate at the third level of the Clifford hierarchy, is realized exactly by a single Dehn twist in Dijkgraaf-Witten theory with finite gauge group $\mathbb{Z}_4$ and generating 3-cocycle $\omega_1 \in H^3(\mathbb{Z}_4, U(1))$, where the non-Clifford phase is described directly from the cocycle data. We compare our constructions across the two frameworks, highlighting how both the Toffoli and T gates rely on similar logical encodings that embed the qubit in the lowest two representations of the theory. A notable contrast emerges from the Dehn twist, which implements the same geometric operation in both theories, but produces a Clifford gate in Chern-Simons theory and a non-Clifford gate in Dijkgraaf-Witten theory, reflecting how different algebraic data determine the topological phases.

These results extend the theory of topological quantum resources initiated in \cite{Salton:2016qpp}, where stabilizer states are prepared by path integration in abelian $U(1)$ Chern-Simons theory, and continued in \cite{Munizzi:2025suf}, where Clifford orbits of $n$-qudit Dicke states are realized topologically in $SU(2)_1$. Related work~\cite{Balasubramanian:2025kaf,Cummings:2025zfe} has also established the topological features of entanglement in Chern-Simons theory and fibered link states. Here we move further beyond the Clifford regime, establishing that topological path integrals can produce non-Clifford gates with quantifiable magic and non-local magic content. Our construction complements that of~\cite{Fliss:2021}, where the magic of topologically prepared knot and link states is quantified in $SU(2)_k$ using the mana monotone at the state level. Conversely, we build our program at the level of unitary operators, constructing magic gates and characterizing their ability to generate non-local, non-stabilizing action.

We discover an important structural distinction between the topological constructions of non-stabilizer operators in Chern-Simons and Dijkgraaf-Witten theories. In Chern-Simons theory, the modular $T$ transformation only produces a Clifford gate (the phase gate), and the non-Clifford properties of the Ising gate $U(\theta)$ arise instead from the continuous parameter $\theta$, which varies between a fixed set of Clifford points. In Dijkgraaf-Witten theory, the modular $T$ action is manifestly non-Clifford, with the magic encoded directly in the 3-cocycle $\omega_1$. This observation suggests that the cohomological data of the topological action constrains the level of the Clifford hierarchy accessible to the theory, an observation that would be interesting to make precise in the language of the cochain-Clifford hierarchy correspondence studied~\cite{CuiGottesmanKrishna2017} on $\mathbb{Z}_2^n$.

The relationship between the algebraic magic constructed in this work and the braiding magic accessible in higher-level theories remains to be studied. For $SU(2)_k$ with $k \geq 3$ and $k \neq 4$, non-abelian braiding enables universal quantum computation~\cite{Freedman2003,Nayak2008}, and the associated braid gates are topologically protected against local perturbations. The magic gates we construct in $SU(2)_1$ do not share this topological protection since they arise from fusion tensor contractions and manifold products rather than anyon braiding. It is unclear whether our topological construction over the disjoint union of manifolds admits a natural anyonic interpretation. Furthermore, how the magic content of unitaries prepared by path integration compares with braiding-based magic at higher levels is an open question, with implications for realizing magic resources in topological quantum computing.

A related open question concerns the topological invariance of magic measures themselves. The non-stabilizing power $m_p = 1/5\sin^2(2\theta)$ depends on the coupling parameter $\theta$ rather than any topological invariant of the underlying manifold. Whether or not a genuinely topological magic measure exists, stable under homeomorphisms and solely dependent on the topology of the manifold and Wilson loop insertion, remains unknown. If such a measure were to exist, it would provide an operator-level analog to the mana computed for knot and link states in~\cite{Fliss:2021}.

In~\cite{Beverland2016} Beverland et al. showed that locality-preserving logical gates in topological codes must lie within a finite level of the generalized Clifford hierarchy, where the precise level is determined by the anyons. Our Dijkgraaf-Witten construction provides a concrete example where topological data, the 3-cocycle, directly determines a gate at the third-level of the Clifford hierarchy, consistent with this general framework. It would be interesting to understand how the open problems identified in our $SU(2)_3$ Toffoli construction, specifically the leakage cancellation and explicit surgery presentation, interact with the constraints derived in~\cite{Beverland2016} for non-abelian models.

Finally, the logical encodings used across our three constructions, $\ket{0}_L = \ket{j=0}$ and $\ket{1}_L = \ket{j=1/2}$ in $SU(2)_k$, and $\ket{0}_L = \ket{j=0}$ and $\ket{1}_L = \ket{j=1}$ in $\mathbb{Z}_4$ Dijkgraaf-Witten theory, embed the logical qubit in the lowest two representations. Whether a single topological framework exists that can realize the Ising gate, the Toffoli gate, and the T gate within a unified encoding, or whether each non-Clifford operation necessarily requires a different topological theory, remains an open problem. A resolution to this question would clarify the extent to which the full Clifford hierarchy can be realized topologically.

\section*{Acknowledgments}

The authors thank Brian Swingle and Anatoly Dymarsky for helpful conversations. W.M. is supported by the Department of Energy (DOE) Office of Science (SC) Grant No DOE DE-FOA-0003432 and by Grant No GBMF12976 of the Gordon and Betty Moore Foundation.

\bibliographystyle{JHEP}

\bibliography{TopoNLMagic}

\end{document}